\documentclass[11pt]{amsart}
\usepackage[margin=1in]{geometry}
\usepackage{amssymb,xcolor,colortbl,diagbox,comment,graphicx,mathpazo,cite}
\usepackage{MnSymbol}
\usepackage{color,float,fullpage,mathrsfs,mathtools}
\usepackage{enumerate,enumitem}
\usepackage[colorlinks=true, pdfstartview=FitV, linkcolor=blue, citecolor=blue, urlcolor=blue]{hyperref}
\usepackage{multirow,array}
\usepackage{makecell}

\def\bpm{\begin{pmatrix}}
\def\epm{\end{pmatrix}}
\def\bvm{\begin{vmatrix}}
\def\evm{\end{vmatrix}}

\def\b#1{\overline{#1}}

\def\XXint#1#2#3{{\setbox0=\hbox{$#1{#2#3}{\int}$}
     \vcenter{\hbox{$#2#3$}}\kern-.5\wd0}}

\newtheorem{proposition}{Proposition}
\newtheorem{theorem}{Theorem}
\newtheorem{corollary}{Corollary}
\newtheorem{lemma}{Lemma}
\theoremstyle{definition}

\makeatletter
\@addtoreset{equation}{section}
\makeatother

\title[Coupled Lakshmanan-Porsezian-Daniel equation]{Inverse scattering transform for the coupled Lakshmanan-Porsezian-Daniel equation with nonzero boundary conditions}

\author{Peng-Fei Han}
\address[P. F. Han]{}
\email{hanpf1995@163.com}
\author{Ru-Suo Ye}
\address[R. S. Ye]{}
\email{rusuoye@163.com}
\author{Yi Zhang$^{*}$}
\address[Corresponding author: Y. Zhang]{Department of Mathematics\\ Zhejiang Normal University\\ Jinhua 321004\\ People's Republic of China.}
\email{zy2836@163.com}

\thanks{$*$ Corresponding author. E-mail address: zy2836@163.com(Yi Zhang)}

\dedicatory{Department of Mathematics, Zhejiang Normal University, Jinhua 321004, People's Republic of China.}


\date{\today}

\begin{document}

\begin{abstract}
The challenge of solving the initial value problem for the coupled Lakshmanan Porsezian Daniel equation, while considering nonzero boundary conditions at infinity, is addressed through the development of a suitable inverse scattering transform. Analytical properties of the Jost eigenfunctions are examined, along with the analysis of scattering coefficient characteristics. This analysis leads to the derivation of additional auxiliary eigenfunctions necessary for the comprehensive investigation of the fundamental eigenfunctions. Two symmetry conditions are discussed to study the eigenfunctions and scattering coefficients. These symmetry results are utilized to rigorously define the discrete spectrum and ascertain the corresponding symmetries of scattering datas. The inverse scattering problem is formulated by the Riemann-Hilbert problem. Then we can derive the exact solutions by coupled Lakshmanan Porsezian Daniel equation, the novel soliton solutions are derived and examined in detail.
\end{abstract}
\maketitle

\tableofcontents


\section{Introduction}
\label{s:intro}

The inverse scattering transform (IST) was first proposed by Gardner, Greene, Kruskal and Miura to exactly analyze the initial-value problems for the famous Korteweg-de Vries equation with Lax pairs in 1967~\cite{1}. Since then, a lot of work has been done to extend this method to other integrable nonlinear systems characterized by Lax pairs~\cite{2}. In certain instances, these equations possess complete integrability, allowing the initial value problem to be potentially solvable through the application of the IST. The nonlinear Schr\"odinger (NLS) equation is a typical model for describing the propagation of nonlinear pulses in monomode fiber~\cite{3}
\begin{equation}\label{NLS}
\begin{split}
\mathrm{i}q_{t}+q_{xx}+2\beta|q|^{2}q=0,
\end{split}
\end{equation}
both in the focusing case $\beta=1$ and defocusing case $\beta=-1$. The NLS equation~\eqref{NLS} is generated as the envelope equation of the dispersive wave system, which is nearly monochromatic and weakly nonlinear in optical fiber system~\cite{4}. The IST for the defocusing NLS equation~\cite{B1} and the focusing NLS equation~\cite{B2} with nonzero boundary conditions (NBCS) at infinity were presented, and the focusing and defocusing Hirota equations were also studied~\cite{H1}.

The present work aims to study the coupled Lakshmanan Porsezian Daniel (LPD) equation~\cite{5}
\begin{equation}\label{LPD:0}
\begin{split}
\mathrm{i}\widetilde{\mathbf{q}}_{t}&+\widetilde{\mathbf{q}}_{xx}
+2\Vert\widetilde{\mathbf{q}}\Vert^{2}\widetilde{\mathbf{q}}+\sigma[\widetilde{\mathbf{q}}_{xxxx}
+2\Vert\widetilde{\mathbf{q}}_{x}\Vert^{2}\widetilde{\mathbf{q}}
+2(\widetilde{\mathbf{q}}_{x}^{\dagger}\widetilde{\mathbf{q}})\widetilde{\mathbf{q}}_{x}
+6(\widetilde{\mathbf{q}}^{\dagger}\widetilde{\mathbf{q}}_{x})\widetilde{\mathbf{q}}_{x}
+4\Vert\widetilde{\mathbf{q}}\Vert^{2}\widetilde{\mathbf{q}}_{xx} \\
&+4(\widetilde{\mathbf{q}}^{\dagger}\widetilde{\mathbf{q}}_{xx})\widetilde{\mathbf{q}}
+2(\widetilde{\mathbf{q}}_{xx}^{\dagger}\widetilde{\mathbf{q}})\widetilde{\mathbf{q}}
+6\Vert\widetilde{\mathbf{q}}\Vert^{4}\widetilde{\mathbf{q}}]=\mathbf{0},
\end{split}
\end{equation}
which can be used to describe the propagation of ultrashort pulse in the birefringent or two-mode fiber, where $\widetilde{\mathbf{q}}=\widetilde{\mathbf{q}}(x,t)=(\widetilde{q}_{1}(x,t),\widetilde{q}_{2}(x,t))^{T}$, the parameter $\sigma$ is a real constant. It is convenient to consider the coupled LPD equation with the NZBC, namely
\begin{equation}\label{1.1}
\begin{split}
\mathrm{i}\mathbf{q}_{t}&+\mathbf{q}_{xx}+2(\Vert\mathbf{q}\Vert^{2}-q_{0}^{2})\mathbf{q}
+\sigma[\mathbf{q}_{xxxx}+2\Vert\mathbf{q}_{x}\Vert^{2}\mathbf{q}
+2(\mathbf{q}_{x}^{\dagger}\mathbf{q})\mathbf{q}_{x}
+6(\mathbf{q}^{\dagger}\mathbf{q}_{x})\mathbf{q}_{x}+4\Vert\mathbf{q}\Vert^{2}\mathbf{q}_{xx} \\
&+4(\mathbf{q}^{\dagger}\mathbf{q}_{xx})\mathbf{q}+2(\mathbf{q}_{xx}^{\dagger}\mathbf{q})\mathbf{q}
+6(\Vert\mathbf{q}\Vert^{4}-q_{0}^{4})\mathbf{q}]=\mathbf{0},
\end{split}
\end{equation}
obtained via the simple change of dependent variable $\widetilde{\mathbf{q}}(x,t)=\mathbf{q}(x,t)\mathrm{e}^{2\mathrm{i}q_{0}^{2}t+6\mathrm{i}\sigma q_{0}^{4}t}$. The corresponding NBCS at infinity are
\begin{equation}\label{1.2}
\begin{split}
\lim_{x\rightarrow\pm\infty}{\mathbf{q}}(x,t)=\mathbf{q}_{\pm}=\mathbf{q}_{0}
\mathbf{e}^{\mathrm{i}\delta_{\pm}},
\end{split}
\end{equation}
where $\mathbf{q}=\mathbf{q}(x,t)=(q_{1}(x,t),q_{2}(x,t))^{T}$ and $\mathbf{q}_{0}$ are two-component vectors, $q_{0}=\Vert \mathbf{q}_{0} \Vert$, with $\delta_{\pm}$ are real numbers. Based on the $N$-fold Darboux transformation, the higher-order rogue waves and multi-breather of the coupled LPD equation~\eqref{1.1} were constructed~\cite{5}. One soliton, bound-state two soliton solutions~\cite{6} and the conservation laws of the coupled LPD equation~\eqref{1.1} were obtained with the aid of the Darboux transformation.

When $q_{2}(x,t)=0$, it is obviously that Eq.~\eqref{1.1} may be reduced into the higher-order NLS equation~\cite{7}:
\begin{equation}\label{1.4}
\begin{split}
\mathrm{i}q_{t}+q_{xx}+2|q|^{2}q+\sigma(q_{xxxx}+4|q_{x}|^{2}q+6q_{x}^{2}q^{*}+8q_{xx}|q|^{2}
+2q^{2}q_{xx}^{*}+6|q|^{4}q)=0,
\end{split}
\end{equation}
with $q=q(x,t)$ denotes the complex envelope and the parameter $\sigma$ is a real constant. The breather solutions and interaction solutions for the higher-order NLS equation~\cite{7} were derived, and these three types of solutions can be transformed into the nonpulsating soliton solutions.

In 1973, Zakharov and Shabat extended the application of IST to some NBCS including the vector NLS equations~\cite{8,9} and made remarkable progress. The vector NLS equations~\cite{B3} serve as universal frameworks to describe the evolution of weakly nonlinear dispersive wave trains. They describe the change of the wave function of the particle with time and space and the energy of the particle~\cite{B4}, the solutions can be used to describe the trajectory of the particle. Shortly thereafter, the IST technique for the defocusing vector NLS equation with NBCS at infinity was completely developed and some mixed solutions of dark solitons and bright solitons were presented in~\cite{10,B5}. The IST for square matrix NLS system~\cite{11} with NBCS was also accomplished. The Riemann-Hilbert approach for higher-order nonlinear evolution equations~\cite{H2,H3} with zero boundary conditions at infinity were also accomplished. By establishing an appropriate IST, the initial value problem of focusing~\cite{12} and defocusing~\cite{13} Manakov systems with NBCS at infinity were studied. Dark-dark solitons, bright-bright, breather-breather and their interactions of the coupled Gerdjikov-Ivanov equation~\cite{B6} with NBCS were presented based on the IST.

With advancements over the past fifteen years, the IST has experienced a notable resurgence, giving rise to a proliferation of research papers and practical implementations~\cite{14}. The research also extends to studying the complex physical characteristics of plasma physics~\cite{H4}, solid-state physics~\cite{H5}, quantum mechanics~\cite{H6}, and their relationships with rogue waves~\cite{17,18,H7}. This paper aims to extend the IST to the coupled LPD equation and study the analytical properties and asymptotic behavior of the Jost eigenfunction. A framework is established for obtaining reflectionless solutions corresponding to any number of simple zeros of the analytical scattering coefficient. For the majority of partial differential equations, solitons correspond to the zeros of the analytical scattering coefficient. Here, we utilize the precise formulation of the IST mentioned earlier to establish a coherent approach for deriving solutions associated with the simple zero of the analytical scattering coefficients. We derive the simple poles within the context of the Riemann-Hilbert (RH) problem and construct these soliton solutions using the precise formulation of the previously mentioned IST.

The process is arranged as follows. In Section 2, we formulate the Jost solutions, scattering matrix, auxiliary eigenfunctions and symmetries. In Section 3, the discrete spectrum and symmetries of the norming constants and the asymptotic behavior are studied by calculation. In Section 4, we formulate the RH problem and trace formulae. In Section 5, we obtain the exact soliton solutions under the condition of reflectionless potentials, and Section 6 contains a final discussion.

\section{Direct scattering problem}
\label{s:Dsp}

\subsection{Uniformization}
\label{s:Dsp:uniform}

The coupled LPD equation~\eqref{1.1} possess Lax pairs representation, which can be expressed as follows:
\begin{equation}\label{2.1}
\begin{split}
\psi_{x}&=\mathbf{A}(x,t,k)\psi, \quad \psi_{t}=\mathbf{E}(x,t,k)\psi,
\end{split}
\end{equation}
while $\psi=\psi(x,t)$ is the vector eigenfunction, the matrices $\mathbf{A}(x,t,k)$ and $\mathbf{E}(x,t,k)$ are written as
\begin{equation}\label{2.2}
\begin{split}
\mathbf{A}=\mathbf{A}(x,t,k)&=-\mathrm{i}k\mathbf{J}+\mathbf{Q}, \\
\mathbf{E}=\mathbf{E}(x,t,k)&=2k[\sigma(\mathbf{Q}_{x}\mathbf{Q}-\mathbf{Q}\mathbf{Q}_{x})
-2\sigma\Vert\mathbf{q}\Vert^{2}\mathbf{Q}-\sigma\mathbf{Q}_{xx}-\mathbf{Q}]
+\mathrm{i}q_{0}^{2}\mathbf{J}+3\mathrm{i}\sigma q_{0}^{4}\mathbf{J}  \\
&+4\mathrm{i}k^{2}[\mathbf{J}_{1}+\sigma\mathbf{J}(\mathbf{Q}_{x}-\mathbf{Q}^{2})]
+8\sigma k^{3}\mathbf{Q}-8\mathrm{i}\sigma k^{4}\mathbf{J}
+3\mathrm{i}\sigma\mathbf{J}(\mathbf{Q}_{x}\mathbf{Q}^{2}+\mathbf{Q}^{2}\mathbf{Q}_{x}) \\
&+\mathrm{i}\sigma\mathbf{J}(\mathbf{Q}_{xx}\mathbf{Q}+\mathbf{Q}\mathbf{Q}_{xx}
-\mathbf{Q}_{xxx}-\mathbf{Q}_{x}^{2})
+\mathrm{i}[1+3\sigma\Vert\mathbf{q}\Vert^{2}]\mathbf{J}\mathbf{Q}^{2}
-\mathrm{i}\mathbf{J}\mathbf{Q}_{x},
\end{split}
\end{equation}
with
\begin{equation}\label{2.3}
\begin{split}
\mathbf{J}=\begin{bmatrix}
    1 & \mathbf{0_{1\times2}}  \\
    \mathbf{0_{2\times1}} & -\mathbf{I_{2\times2}}
  \end{bmatrix}, \quad
\mathbf{J}_{1}=\begin{bmatrix}
    1 & \mathbf{0_{1\times2}}  \\
    \mathbf{0_{2\times1}} & \mathbf{0_{2\times2}}
  \end{bmatrix}, \quad
\mathbf{Q}=\mathbf{Q}(x,t)=\begin{bmatrix}
    0 & \mathbf{q}^{\dagger}   \\
    -\mathbf{q} & \mathbf{0_{2\times2}}
  \end{bmatrix},
\end{split}
\end{equation}
where $k$ is the spectral parameter, and $\dagger$ stands for the Hermitian transpose. It can be readily determined that the coupled LPD equation~\eqref{1.1} can be derived through the examination of the following condition~\cite{A1}
\begin{equation}\label{2.4}
\begin{split}
\mathbf{A}_{t}-\mathbf{E}_{x}+[\mathbf{A},\mathbf{E}]=\mathbf{0},
\end{split}
\end{equation}
the compatibility condition of $\psi_{xt}=\psi_{tx}$ can be confirmed by performing direct calculations and observing that $\mathbf{J}\mathbf{Q}=-\mathbf{Q}\mathbf{J}$.

To define the Jost eigenfunctions as $x\to\pm\infty$, the spatial and temporal evolution of the solutions of the Lax pairs will be asymptotically
\begin{equation}\label{2.5}
\begin{split}
\psi_{x}=\mathbf{A}_{\pm}\psi, \quad \psi_{t}=\mathbf{E}_{\pm}\psi,
\end{split}
\end{equation}
with
\begin{equation}\label{2.6}
\begin{split}
\lim_{x\rightarrow\pm\infty}{\mathbf{A}}=\mathbf{A}_{\pm}&=-\mathrm{i}k\mathbf{J}+\mathbf{Q}_{\pm}, \\
\lim_{x\rightarrow\pm\infty}{\mathbf{E}}=\mathbf{E}_{\pm}&=
(-8\mathrm{i}\sigma k^{4}+\mathrm{i}q_{0}^{2}+3\mathrm{i}\sigma q_{0}^{4})\mathbf{J}
+8\sigma k^{3}\mathbf{Q}_{\pm}+4\mathrm{i}k^{2}(\mathbf{J}_{1}-\sigma\mathbf{J}\mathbf{Q}_{\pm}^{2})\\
&-2k(2\sigma q_{0}^{2}\mathbf{Q}_{\pm}+\mathbf{Q}_{\pm})+\mathrm{i}(1+3\sigma q_{0}^{2})\mathbf{J}\mathbf{Q}_{\pm}^{2}.
\end{split}
\end{equation}
The eigenvalues of $\mathbf{A}_{\pm}$ and $\mathbf{E}_{\pm}$ are as follows
\begin{equation}\label{2.7}
\begin{split}
\mathbf{A}_{\pm,1}&=\mathrm{i}k, \quad \mathbf{E}_{\pm,1}=8\mathrm{i}k^{4}\sigma-\mathrm{i}q_{0}^{2}-3\mathrm{i}\sigma q_{0}^{4}, \\
\mathbf{A}_{\pm,2,3}&=\pm\mathrm{i}\lambda, \quad
\mathbf{E}_{\pm,2,3}=2\mathrm{i}k^{2} \pm 2\mathrm{i}\lambda(k-4k^{3}\sigma+2k\sigma q_{0}^{2}).
\end{split}
\end{equation}

Given the doubly branched nature of the eigenvalues, we introduce the two-sheeted Riemann surface so that $\lambda(k)$ a single-valued function~\cite{B2}
\begin{equation}\label{2.8}
\begin{split}
\lambda^{2}=k^{2}+q_{0}^{2}.
\end{split}
\end{equation}
Letting
\begin{equation}\label{2.9}
\begin{split}
k+\mathrm{i}q_{0}=\kappa_{1}\mathrm{e}^{\mathrm{i}\vartheta_{1}}, \quad
k-\mathrm{i}q_{0}=\kappa_{2}\mathrm{e}^{\mathrm{i}\vartheta_{2}}, \quad
-\frac{\pi}{2}<\vartheta_{1},\vartheta_{2}<\frac{3\pi}{2},
\end{split}
\end{equation}
two single-valued analytic functions are obtained
\begin{equation}\label{2.10}
\begin{split}
\lambda(k)=\left\{\begin{array}{ll}
\sqrt{\kappa_{1}\kappa_{2}}\, \mathrm{e}^{\mathrm{i}(\vartheta_{1}+\vartheta_{2})/2}, & \text { on } S_{1}, \\
-\sqrt{\kappa_{1}\kappa_{2}}\, \mathrm{e}^{\mathrm{i}(\vartheta_{1}+\vartheta_{2})/2}, & \text { on } S_{2},
\end{array}\right.
\end{split}
\end{equation}
the related properties of the two-sheeted Riemann surface can be referred to~\cite{B2}.

Next, we proceed to define the uniformization variable
\begin{equation}\label{2.11}
\begin{split}
\lambda=z-k,
\end{split}
\end{equation}
the corresponding inverse map is given by
\begin{equation}\label{2.12}
\begin{split}
k=\frac{1}{2}(z-\frac{q_{0}^{2}}{z}), \quad \lambda=\frac{1}{2}(z+\frac{q_{0}^{2}}{z}).
\end{split}
\end{equation}
Denoting
\begin{equation}\label{2.15}
\begin{split}
U^{+}&=\left\{z \in \mathbb{C}:\left(|z|^{2}-q_{o}^{2}\right) \mathfrak{Im} z>0\right\}, \quad
U^{-}=\left\{z \in \mathbb{C}:\left(|z|^{2}-q_{o}^{2}\right) \mathfrak{Im} z<0\right\}.
\end{split}
\end{equation}

\subsection{Jost eigenfunctions, scattering matrix and analyticity}
\label{s:Dsp:Jost}

On the complex $z$-plane, the continuous spectrum consists is $\Sigma=\mathbb{R}\cup C_{0}$, where $C_{0}$ is a circle with the origin as the center and the radius $q_{0}$ in the complex $z$-plane. Define the following properties for any two-component vector
\begin{equation}\label{2.16}
\begin{split}
\mathbf{v}=(v_{1},v_{2})^{T}, \quad  \mathbf{v}^{\perp}=(v_{2},-v_{1})^{\dagger}.
\end{split}
\end{equation}

The eigenvector matrix is given below
\begin{equation}\label{2.17}
\begin{split}
\mathbf{Y}_{\pm}=\mathbf{Y}_{\pm}(z)=\begin{bmatrix}
    1 & 0 &  -\displaystyle{\frac{\mathrm{i}q_{0}}{z}}  \\
    -\displaystyle{\frac{\mathrm{i}}{z}\mathbf{q}_{\pm}} &
    \displaystyle{\frac{\mathbf{q}_{\pm}^{\perp}}{q_{0}}} &
     \displaystyle{\frac{\mathbf{q}_{\pm}}{q_{0}}}
  \end{bmatrix}, \quad  \operatorname{det}\mathbf{Y}_{\pm}(z)=1+\frac{q_{0}^{2}}{z^{2}}:=\rho(z),
\end{split}
\end{equation}
where
\begin{equation}\label{2.18}
\begin{split}
 \quad
\mathbf{Y}_{\pm}^{-1}(z)=\frac{1}{\rho(z)}\begin{bmatrix}
    1 & \displaystyle{\frac{\mathrm{i}\mathbf{q}_{\pm}^{\dagger}}{z}}  \\
    0 & \displaystyle{\frac{\rho(z)}{q_{0}}(\mathbf{q}_{\pm}^{\perp})^{\dagger}}  \\
    \displaystyle{\frac{\mathrm{i}q_{0}}{z}} &
    \displaystyle{\frac{\mathbf{q}_{\pm}^{\dagger}}{q_{0}}}  \\
  \end{bmatrix}.
\end{split}
\end{equation}

It can be directly computed that $\mathbf{A}_{\pm}$ and $\mathbf{E}_{\pm}$~\eqref{2.6} satisfy the following conditions
\begin{equation}\label{2.19}
\begin{split}
\mathbf{A}_{\pm}\mathbf{Y}_{\pm}&=\mathrm{i}\mathbf{Y}_{\pm}\mathbf{\Lambda}_{1}, \quad
\mathbf{E}_{\pm}\mathbf{Y}_{\pm}=-\mathrm{i}\mathbf{Y}_{\pm}\mathbf{\Lambda}_{2},
\end{split}
\end{equation}
where
\begin{equation}\label{2.20}
\begin{split}
\mathbf{\Lambda}_{1}&=\operatorname{diag}(-\lambda,k,\lambda), \\
\mathbf{\Lambda}_{2}&=\operatorname{diag}\left(-2\lambda(k-4k^{3}\sigma+2k\sigma q_{0}^{2})-2k^{2}, 3\sigma q_{0}^{4}+q_{0}^{2}-8k^{4}\sigma, 2\lambda(k-4k^{3}\sigma+2k\sigma q_{0}^{2})-2k^{2} \right),
\end{split}
\end{equation}
and satisfy $[\mathbf{A}_{\pm},\mathbf{E}_{\pm}]=\mathbf{0}$, the Jost eigenfunctions $\psi(x,t,z)$ of the Lax pairs on $z\in\Sigma$ satisfying the boundary conditions
\begin{equation}\label{2.21}
\begin{split}
\psi_{\pm}=\psi_{\pm}(x,t,z)=\mathbf{Y}_{\pm}(z) \mathbf{e}^{\mathrm{i}\mathbf{\Omega}(x,t,z)}+o(1), \quad x\rightarrow\pm\infty,
\end{split}
\end{equation}
with
\begin{equation}\label{2.22}
\begin{split}
\mathbf{\Omega}&=\mathbf{\Omega}(x,t,z)=\operatorname{diag}(\delta_{1},\delta_{2},\delta_{3}), \quad
\delta_{1}=-\lambda x+2(k^{2}+k\lambda-4k^{3}\lambda\sigma+2k\lambda\sigma q_{0}^{2})t, \\
\delta_{2}&=kx+(8k^{4}\sigma-3\sigma q_{0}^{4}-q_{0}^{2})t, \quad
\delta_{3}=\lambda x+2(k^{2}-k\lambda+4k^{3}\lambda\sigma-2k\lambda\sigma q_{0}^{2})t.
\end{split}
\end{equation}

As usual, we introduce modified eigenfunctions
\begin{equation}\label{2.23}
\begin{split}
\nu_{\pm}=\nu_{\pm}(x,t,z)=\psi_{\pm} \mathbf{e}^{-\mathrm{i}\mathbf{\Omega}},
\end{split}
\end{equation}
so that
\begin{equation}\label{2.24}
\begin{split}
\lim_{x\rightarrow\pm\infty}{\nu_{\pm}}=\mathbf{Y}_{\pm}.
\end{split}
\end{equation}
We rewrite the Lax pairs~\eqref{2.1} as
\begin{equation}\label{2.25}
\begin{split}
(\psi_{\pm})_{x}&=\mathbf{A}_{\pm}\psi_{\pm}+\Delta\mathbf{A}_{\pm}\psi_{\pm}, \quad
(\psi_{\pm})_{t}=\mathbf{E}_{\pm}\psi_{\pm}+\Delta\mathbf{E}_{\pm}\psi_{\pm},
\end{split}
\end{equation}
where $\Delta\mathbf{A}_{\pm}=\mathbf{A}-\mathbf{A}_{\pm}$ and $\Delta\mathbf{E}_{\pm}=\mathbf{E}-\mathbf{E}_{\pm}$, the Lax pairs~\eqref{2.25} can be written
as
\begin{subequations}\label{2.26}
\begin{align}
(\mathbf{Y}_{\pm}^{-1}\nu_{\pm})_{x}&=
[\mathrm{i}\mathbf{\Omega}_{x},\mathbf{Y}_{\pm}^{-1}\nu_{\pm}]
+\mathbf{Y}_{\pm}^{-1}\Delta\mathbf{A}_{\pm}\nu_{\pm} ,  \\
(\mathbf{Y}_{\pm}^{-1}\nu_{\pm})_{t}&=
[\mathrm{i}\mathbf{\Omega}_{t},\mathbf{Y}_{\pm}^{-1}\nu_{\pm}]
+\mathbf{Y}_{\pm}^{-1}\Delta\mathbf{E}_{\pm}\nu_{\pm},
\end{align}
\end{subequations}
where $\mathbf{\Omega}_{x}=\mathbf{\Lambda}_{1}$ and $\mathbf{\Omega}_{t}=-\mathbf{\Lambda}_{2}$. The system~\eqref{2.26} can be written in the full derivative form
\begin{equation}\label{2.27}
\begin{split}
\mathrm{d}[\mathbf{e}^{-\mathrm{i}\mathbf{\Omega}}\mathbf{Y}_{\pm}^{-1}\nu_{\pm}
\mathbf{e}^{\mathrm{i}\mathbf{\Omega}}]&=
\mathbf{e}^{-\mathrm{i}\mathbf{\Omega}}[\sigma_{1}\mathrm{d}x+\sigma_{2}\mathrm{d}t]
\mathbf{e}^{\mathrm{i}\mathbf{\Omega}},
\end{split}
\end{equation}
with
\begin{equation}\label{2.28}
\begin{split}
\sigma_{1}=\mathbf{Y}_{\pm}^{-1}\Delta\mathbf{A}_{\pm}\nu_{\pm}, \quad
\sigma_{2}=\mathbf{Y}_{\pm}^{-1}\Delta\mathbf{E}_{\pm}\nu_{\pm}.
\end{split}
\end{equation}

It is verified that the spectral problem about $\nu_{\pm}(x,t,z)$ is equivalent to the Volterra integral equations
\begin{equation}\label{2.29}
\begin{split}
\nu_{\pm}(x,t,z)&=\mathbf{Y}_{\pm}(z)+\int_{\pm\infty}^{x}\mathbf{Y}_{\pm}(z)
\mathbf{e}^{\mathrm{i}(x-r)\mathbf{\Lambda}_{1}(z)}\sigma_{3}
\mathbf{e}^{-\mathrm{i}(x-r)\mathbf{\Lambda}_{1}(z)}\mathrm{d}r,
\end{split}
\end{equation}
with
\begin{equation}\label{2.30}
\begin{split}
\sigma_{3}&=\mathbf{Y}_{\pm}^{-1}(z)\Delta\mathbf{A}_{\pm}(r,t)\nu_{\pm}(r,t,z).
\end{split}
\end{equation}
Through a rigorous approach, the Jost solutions can be precisely defined as the solutions to the Volterra integral equations~\eqref{2.29}. The following two theorems can be proved by referring to~\cite{12}.

\begin{theorem}\label{thm:1}
Suppose that $\mathbf{q}(\cdot, t)-\mathbf{q}_{-} \in L^{1}(-\infty, a)$ $(\mathbf{q}(\cdot, t)-\mathbf{q}_{+} \in L^{1}(a, \infty))$ for any constant $a\in\mathbb{R}$, the subsequent columns of $\nu_{-}(x,t,z)$ $(\nu_{+}(x,t,z))$ satisfy the properties
\begin{subequations}\label{2.31}
\begin{align}
\nu_{-,1} & \;:\; z\in\mathbb{U}_{1}, \quad  \nu_{-,2} \;:\; \mathfrak{Im}z<0, \quad
\nu_{-,3} \;:\; z\in\mathbb{U}_{4},  \\
\nu_{+,1} & \;:\; z\in\mathbb{U}_{2}, \quad  \nu_{+,2} \;:\; \mathfrak{Im}z>0, \quad
\nu_{+,3} \;:\; z\in\mathbb{U}_{3},
\end{align}
\end{subequations}
where $\mathbb{U}_{1}$, $\mathbb{U}_{2}$, $\mathbb{U}_{3}$ and $\mathbb{U}_{4}$ are
\begin{subequations}\label{2.32}
\begin{align}
\mathbb{U}_{1}&\equiv\{ \,z\;|\; \mathfrak{Im}z>0  \; \text{and} \; |z|>q_{0} \, \}, \quad
\mathbb{U}_{2}\equiv\{ \,z\;|\; \mathfrak{Im}z<0 \; \text{and} \; |z|>q_{0} \, \},  \\
\mathbb{U}_{3}&\equiv\{ \,z\;|\; \mathfrak{Im}z<0 \; \text{and} \; |z|<q_{0} \, \}, \quad
\mathbb{U}_{4}\equiv\{ \,z\;|\; \mathfrak{Im}z>0 \; \text{and} \; |z|<q_{0} \, \}.
\end{align}
\end{subequations}
\end{theorem}

Equation~\eqref{2.23} suggests that the columns of $\psi_{\pm}$ also possess the same analyticity and boundedness. If $\psi$ is an arbitrary solution of the Lax pairs~\eqref{2.1}, we have
\begin{equation}\label{2.35}
\begin{split}
\partial_{x}[\operatorname{det}\psi_{\pm}]&=\operatorname{tr}\mathbf{A} \operatorname{det}\psi_{\pm}, \quad
\partial_{t}[\operatorname{det}\psi_{\pm}]=\operatorname{tr}\mathbf{E} \operatorname{det}\psi_{\pm},
\end{split}
\end{equation}
where
\begin{equation}\label{2.36}
\begin{split}
\operatorname{tr}\mathbf{A}=\mathrm{i}k, \quad
\operatorname{tr}\mathbf{E}=\mathrm{i}(8k^{4}\sigma+4k^{2}-3\sigma q_{0}^{4}-q_{0}^{2}).
\end{split}
\end{equation}
Abel's theorem yields
\begin{equation}\label{2.37}
\begin{split}
\partial x[\operatorname{det}\nu_{\pm}]&=
\partial x[\operatorname{det}(\psi_{\pm}\mathbf{e}^{-\mathrm{i}\mathbf{\Omega}})]=0, \quad
\partial t[\operatorname{det}\nu_{\pm}]=
\partial t[\operatorname{det}(\psi_{\pm}\mathbf{e}^{-\mathrm{i}\mathbf{\Omega}})]=0.
\end{split}
\end{equation}
Then~\eqref{2.21} implies
\begin{equation}\label{2.38}
\begin{split}
\operatorname{det}\psi_{\pm}=\rho(z)\mathrm{e}^{\mathrm{i}(\delta_{1}
+\delta_{2}+\delta_{3})}, \quad z\in\Sigma \backslash \{\pm\mathrm{i}q_{0}\}.
\end{split}
\end{equation}

The relevant scattering matrix $\mathbf{S}(z)$ and $\mathbf{H}(z)$ can be defined by
\begin{equation}\label{2.39}
\begin{split}
\psi_{-}=\psi_{+}\mathbf{S}(z), \quad \mathbf{H}(z):=\mathbf{S}^{-1}(z), \quad z\in\Sigma \backslash \{\pm\mathrm{i}q_{0}\}.
\end{split}
\end{equation}
Moreover,~\eqref{2.38} and~\eqref{2.39} imply
\begin{equation}\label{2.40}
\begin{split}
\hbox{det} \;\mathbf{S}(z)=1, \quad z\in\Sigma \backslash \{\pm\mathrm{i}q_{0}\}.
\end{split}
\end{equation}

The modified eigenfunctions $\nu_{\pm}(x,t,z)$ are bounded for $x\in\mathbb{R}$ and for all $z$ within the interior of their respective domains of analyticity
\begin{subequations}\label{2.41}
\begin{align}
\nu_{-}&=\mathbf{Y}_{+}\mathbf{e}^{\mathrm{i}\mathbf{\Omega}}\mathbf{S}(z)
\mathbf{e}^{-\mathrm{i}\mathbf{\Omega}}+o(1), \quad (x\rightarrow+\infty), \\
\nu_{+}&=\mathbf{Y}_{-}\mathbf{e}^{\mathrm{i}\mathbf{\Omega}}\mathbf{H}(z)
\mathbf{e}^{-\mathrm{i}\mathbf{\Omega}}+o(1), \quad (x\rightarrow-\infty).
\end{align}
\end{subequations}

\begin{theorem}\label{thm:2}
Under the same hypotheses as in Theorem~\ref{thm:1}, the scattering coefficients satisfy the following properties
\begin{subequations}\label{2.42}
\begin{align}
s_{11}(z) & \;:\; z\in\mathbb{U}_{1}, \quad  s_{22}(z) \;:\; \mathfrak{Im}z<0, \quad
s_{33}(z) \;:\; z\in\mathbb{U}_{4},  \\
h_{11}(z) & \;:\; z\in\mathbb{U}_{2}, \quad  h_{22}(z) \;:\; \mathfrak{Im}z>0, \quad
h_{33}(z) \;:\; z\in\mathbb{U}_{3}.
\end{align}
\end{subequations}
\end{theorem}

\subsection{Adjoint problem}
\label{s:Dsp:Ap}

To establish the scattering problem, it is necessary to have a complete analytic function. We consider the so-called ``adjoint'' Lax pairs~\cite{A3}
\begin{equation}\label{2.43}
\begin{split}
\tilde{\psi}_{x}&=\tilde{\mathbf{A}}\tilde{\psi}, \quad
\tilde{\psi}_{t}=\tilde{\mathbf{E}}\tilde{\psi},
\end{split}
\end{equation}
with
\begin{equation}\label{2.44}
\begin{split}
\tilde{\mathbf{A}}&=\mathrm{i}k\mathbf{J}+\mathbf{Q}^{*}, \quad
\mathbf{Q}^{T}=-\mathbf{Q}^{*}, \quad \mathbf{Q}^{\dagger}=-\mathbf{Q}, \\
\tilde{\mathbf{E}}&=2k[\sigma(\mathbf{Q}_{x}^{*}\mathbf{Q}^{*}-\mathbf{Q}^{*}\mathbf{Q}_{x}^{*})
-2\sigma\Vert\mathbf{q}\Vert^{2}\mathbf{Q}^{*}-\sigma\mathbf{Q}_{xx}^{*}-\mathbf{Q}^{*}]
+\mathrm{i}\mathbf{J}\mathbf{Q}_{x}^{*}-\mathrm{i}q_{0}^{2}\mathbf{J}
-3\mathrm{i}\sigma q_{0}^{4}\mathbf{J}  \\
&-4\mathrm{i}k^{2}[\mathbf{J}_{1}+\sigma\mathbf{J}(\mathbf{Q}_{x}^{*}-(\mathbf{Q}^{*})^{2})]
+8\sigma k^{3}\mathbf{Q}^{*}-8\mathrm{i}\sigma k^{4}\mathbf{J}
-3\mathrm{i}\sigma\mathbf{J}[\mathbf{Q}_{x}^{*}(\mathbf{Q}^{*})^{2}
+(\mathbf{Q}^{*})^{2}\mathbf{Q}_{x}^{*}]  \\
&-\mathrm{i}\sigma\mathbf{J}[\mathbf{Q}_{xx}^{*}\mathbf{Q}^{*}+\mathbf{Q}^{*}\mathbf{Q}_{xx}^{*}
-\mathbf{Q}_{xxx}^{*}-(\mathbf{Q}_{x}^{*})^{2}]
-\mathrm{i}(1+3\sigma\Vert\mathbf{q}\Vert^{2})\mathbf{J}(\mathbf{Q}^{*})^{2}.
\end{split}
\end{equation}

Note that $\tilde{\mathbf{A}}(z)=\mathbf{A}^{*}(z^{*})$ and $\tilde{\mathbf{E}}(z)=\mathbf{E}^{*}(z^{*})$ for all $z\in\Sigma$, while ``$\times$'' denotes the usual cross product. For any number of column vectors $\mathbf{u}_{1},\mathbf{u}_{2}\in \mathbb{C}^{3}$, there are the following properties
\begin{equation}\label{2.46}
\begin{split}
&[(\mathbf{J}\mathbf{u}_{1})\times\mathbf{u}_{2}]+[\mathbf{u}_{1}\times(\mathbf{J}\mathbf{u}_{2})]
+[\mathbf{u}_{1}\times\mathbf{u}_{2}]
+[(\mathbf{J}\mathbf{u}_{1})\times(\mathbf{J}\mathbf{u}_{2})]=\mathbf{0},\\
&\mathbf{J}[\mathbf{u}_{1}\times\mathbf{u}_{2}]
-[(\mathbf{J}\mathbf{u}_{1})\times(\mathbf{J}\mathbf{u}_{2})]=\mathbf{0}, \\
&\mathbf{Q}[\mathbf{u}_{1}\times\mathbf{u}_{2}]+[(\mathbf{Q}^{T}\mathbf{u}_{1})\times\mathbf{u}_{2}]
+[\mathbf{u}_{1}\times(\mathbf{Q}^{T}\mathbf{u}_{2})]=\mathbf{0}, \\
&\mathbf{J}\mathbf{Q}^{2}[\mathbf{u}_{1}\times\mathbf{u}_{2}]
+[(\mathbf{J}(\mathbf{Q}^{T})^{2}\mathbf{u}_{1})\times\mathbf{u}_{2}]
+[\mathbf{u}_{1}\times(\mathbf{J}(\mathbf{Q}^{T})^{2}\mathbf{u}_{2})]=\mathbf{0}.
\end{split}
\end{equation}
Using these identities it is straightforward to prove the following proposition.

\begin{proposition}\label{pro:1}
If $\tilde{\mathbf{v}}_{1}$ and $\tilde{\mathbf{v}}_{2}$ are two arbitrary solutions of the adjoint problem~\eqref{2.43}, then
\begin{equation}\label{2.47}
\begin{split}
\mathbf{v}=\mathrm{e}^{\mathrm{i}(\delta_{1}+\delta_{2}+\delta_{3})}
[\tilde{\mathbf{v}}_{1}\times\tilde{\mathbf{v}}_{2}]
\end{split}
\end{equation}
is a solution of the Lax pairs~\eqref{2.1}.
\end{proposition}

In light of this finding, we utilize the outcome to create four additional analytic eigenfunctions, each corresponding to a fundamental domain. This is accomplished by generating Jost eigenfunctions for the adjoint problem~\eqref{2.43}. As $x\to \pm\infty$, the spatial and temporal evolution of the solutions of the adjoint Lax pairs~\eqref{2.43} will be asymptotically
\begin{equation}\label{2.48}
\begin{split}
\tilde{\psi}_{x}&=\tilde{\mathbf{A}}_{\pm}\tilde{\psi}, \quad \tilde{\psi}_{t}=\tilde{\mathbf{E}}_{\pm}\tilde{\psi},
\end{split}
\end{equation}
with
\begin{equation}\label{2.49}
\begin{split}
\lim_{x\rightarrow\pm\infty}{\tilde{\mathbf{A}}}=\tilde{\mathbf{A}}_{\pm}&=
\mathrm{i}k\mathbf{J}+\mathbf{Q}_{\pm}^{*}, \\
\lim_{x\rightarrow\pm\infty}{\tilde{\mathbf{E}}}=\tilde{\mathbf{E}}_{\pm}&=
-2k(2\sigma q_{0}^{2}+1)\mathbf{Q}^{*}-\mathrm{i}q_{0}^{2}\mathbf{J}
-3\mathrm{i}\sigma q_{0}^{4}\mathbf{J}
-4\mathrm{i}k^{2}[\mathbf{J}_{1}-\sigma\mathbf{J}(\mathbf{Q}^{*})^{2}] \\
&+8\sigma k^{3}\mathbf{Q}^{*}-8\mathrm{i}\sigma k^{4}\mathbf{J}
-\mathrm{i}(1+3\sigma q_{0}^{2})\mathbf{J}(\mathbf{Q}^{*})^{2}.
\end{split}
\end{equation}
The eigenvalues of $\tilde{\mathbf{A}}_{\pm}$ and $\tilde{\mathbf{E}}_{\pm}$ are as follows
\begin{equation}\label{2.50}
\begin{split}
\tilde{\mathbf{A}}_{\pm,1}&=-\mathrm{i}k, \quad \tilde{\mathbf{A}}_{\pm,2,3}=\pm\mathrm{i}\lambda, \\
\tilde{\mathbf{E}}_{\pm,1}&=\mathrm{i}q_{0}^{2}+3\mathrm{i}\sigma q_{0}^{4}-8\mathrm{i}k^{4}\sigma, \quad
\tilde{\mathbf{E}}_{\pm,2,3}=-2\mathrm{i}k^{2} \pm 2\mathrm{i}\lambda(k-4k^{3}\sigma+2k\sigma q_{0}^{2}).
\end{split}
\end{equation}

It can be directly computed that $\tilde{\mathbf{A}}_{\pm}$ and $\tilde{\mathbf{E}}_{\pm}$~\eqref{2.49} satisfy the following conditions
\begin{equation}\label{2.51}
\begin{split}
\tilde{\mathbf{A}}_{\pm}\tilde{\mathbf{Y}}_{\pm}(z)&=
-\mathrm{i}\tilde{\mathbf{Y}}_{\pm}(z)\mathbf{\Lambda}_{1}, \quad
\tilde{\mathbf{E}}_{\pm}\tilde{\mathbf{Y}}_{\pm}(z)=
\mathrm{i}\tilde{\mathbf{Y}}_{\pm}(z)\mathbf{\Lambda}_{2}.
\end{split}
\end{equation}
Note that $\tilde{\mathbf{Y}}_{\pm}(z)=\mathbf{Y}_{\pm}^{*}(z^{*})$ and $\operatorname{det}\tilde{\mathbf{Y}}_{\pm}(z)=\rho(z)$. As before, we define the Jost eigenfunctions of the two parts adjoint Lax pairs~\eqref{2.43}
\begin{equation}\label{2.52}
\begin{split}
\tilde{\psi}_{\pm}=\tilde{\psi}_{\pm}(x,t,z)=\tilde{\mathbf{Y}}_{\pm}(z) \mathbf{e}^{-\mathrm{i}\mathbf{\Omega}}+o(1), \quad x\rightarrow\pm\infty, \quad z\in\Sigma.
\end{split}
\end{equation}

Introducing the modified Jost solutions
\begin{equation}\label{2.53}
\begin{split}
\tilde{\nu}_{\pm}=\tilde{\nu}_{\pm}(x,t,z)=\tilde{\psi}_{\pm} \mathbf{e}^{\mathrm{i}\mathbf{\Omega}},
\end{split}
\end{equation}
the following columns of $\tilde{\nu}_{\pm}(x,t,z)$ satisfy the following properties
\begin{subequations}\label{2.54}
\begin{align}
\tilde{\nu}_{-,1} & \;:\; z\in\mathbb{U}_{2}, \quad
\tilde{\nu}_{-,2} \;:\; \mathfrak{Im}z>0, \quad
\tilde{\nu}_{-,3} \;:\; z\in\mathbb{U}_{3},  \\
\tilde{\nu}_{+,1} & \;:\; z\in\mathbb{U}_{1}, \quad
\tilde{\nu}_{+,2} \;:\; \mathfrak{Im}z<0, \quad
\tilde{\nu}_{+,3} \;:\; z\in\mathbb{U}_{4}.
\end{align}
\end{subequations}

Modified Jost solutions~\eqref{2.53} suggest that the columns of $\tilde{\psi}_{\pm}$ also possess the same properties of analyticity and boundedness. We can introduce the adjoint scattering matrix as
\begin{equation}\label{2.55}
\begin{split}
\tilde{\psi}_{-}=\tilde{\psi}_{+}\tilde{\mathbf{S}}(z), \quad \tilde{\mathbf{H}}(z)=\tilde{\mathbf{S}}^{-1}(z), \quad z\in\Sigma \backslash \{\pm\mathrm{i}q_{0}\},
\end{split}
\end{equation}
the following scattering coefficients satisfy the following properties
\begin{subequations}\label{2.56}
\begin{align}
\tilde{s}_{11}(z) & \;:\; z\in\mathbb{U}_{2}, \quad  \tilde{s}_{22}(z) \;:\; \mathfrak{Im}z>0, \quad \tilde{s}_{33}(z) \;:\; z\in\mathbb{U}_{3},  \\
\tilde{h}_{11}(z) & \;:\; z\in\mathbb{U}_{1}, \quad  \tilde{h}_{22}(z) \;:\; \mathfrak{Im}z<0, \quad \tilde{h}_{33}(z) \;:\; z\in\mathbb{U}_{4}.
\end{align}
\end{subequations}

Next, we can establish the auxiliary eigenfunctions
\begin{subequations}\label{2.57}
\begin{align}
\mu_{1}(z)&=\mathrm{e}^{\mathrm{i}[\delta_{1}+\delta_{2}+\delta_{3}](z)}
[\tilde{\psi}_{+,1}\times\tilde{\psi}_{-,2}](z),  \quad
\mu_{2}(z)=\mathrm{e}^{\mathrm{i}[\delta_{1}+\delta_{2}+\delta_{3}](z)}
[\tilde{\psi}_{-,1}\times\tilde{\psi}_{+,2}](z),  \\
\mu_{3}(z)&=\mathrm{e}^{\mathrm{i}[\delta_{1}+\delta_{2}+\delta_{3}](z)}
[\tilde{\psi}_{+,2}\times\tilde{\psi}_{-,3}](z),  \quad
\mu_{4}(z)=\mathrm{e}^{\mathrm{i}[\delta_{1}+\delta_{2}+\delta_{3}](z)}
[\tilde{\psi}_{-,2}\times\tilde{\psi}_{+,3}](z).
\end{align}
\end{subequations}
Using the results of~\eqref{2.53} and~\eqref{2.54}, the auxiliary eigenfunctions $\mu_{1}(z)$, $\mu_{2}(z)$, $\mu_{3}(z)$ and $\mu_{4}(z)$ are analytic for $z\in\mathbb{U}_{1}$, $z\in\mathbb{U}_{2}$, $z\in\mathbb{U}_{3}$ and $z\in\mathbb{U}_{4}$, respectively.

Note that there is a straightforward connection between the adjoint Jost eigenfunctions and the eigenfunctions of the Lax pairs~\eqref{2.1}.

\begin{corollary}\label{cor:1}
For all cyclic indices $j$, $l$ and $m$ with $z\in\Sigma$,
\begin{subequations}\label{2.58}
\begin{align}
\psi_{\pm,j}(z)&=\frac{\mathrm{e}^{\mathrm{i}[\delta_{1}+\delta_{2}+\delta_{3}]
(z)}}{\rho_{j}(z)}[\tilde{\psi}_{\pm,l}\times\tilde{\psi}_{\pm,m}](z),  \\
\tilde{\psi}_{\pm,j}(z)&=
\frac{\mathrm{e}^{-\mathrm{i}[\delta_{1}+\delta_{2}+\delta_{3}](z)}}{\rho_{j}(z)}
[\psi_{\pm,l}\times\psi_{\pm,m}](z),
\end{align}
\end{subequations}
where
\begin{equation}\label{2.59}
\begin{split}
\rho_{1}(z)=1, \quad \rho_{2}(z)=\rho(z), \quad \rho_{3}(z)=1.
\end{split}
\end{equation}
\end{corollary}

This connection leads to a correlation between the respective scattering matrices, then the scattering matrices $\mathbf{S}(z)$ and $\tilde{\mathbf{S}}(z)$ satisfy the following relationship
\begin{equation}\label{2.60}
\begin{split}
\tilde{\mathbf{S}}^{-1}(z)=\mathbf{J}_{2}(z)\mathbf{S}^{T}(z)
\mathbf{J}_{2}^{-1}(z), \quad  \mathbf{J}_{2}(z)=\operatorname{diag}(1,\rho(z),1).
\end{split}
\end{equation}
The Jost eigenfunctions have the following decompositions in $z\in\Sigma$
\begin{subequations}\label{2.61}
\begin{align}
\psi_{-,1}(z)&=\frac{\mu_{3}(z)+s_{21}(z)\psi_{-,2}(z)}{s_{22}(z)}
=\frac{\mu_{4}(z)+s_{31}(z)\psi_{-,3}(z)}{s_{33}(z)},  \\
\psi_{+,1}(z)&=\frac{\mu_{3}(z)+h_{31}(z)\psi_{+,3}(z)}{h_{33}(z)}
=\frac{\mu_{4}(z)+h_{21}(z)\psi_{+,2}(z)}{h_{22}(z)},  \\
\psi_{-,3}(z)&=\frac{\mu_{2}(z)+s_{23}(z)\psi_{-,2}(z)}{s_{22}(z)}
=\frac{\mu_{1}(z)+s_{13}(z)\psi_{-,1}(z)}{s_{11}(z)},  \\
\psi_{+,3}(z)&=\frac{\mu_{1}(z)+h_{23}(z)\psi_{+,2}(z)}{h_{22}(z)}
=\frac{\mu_{2}(z)+h_{13}(z)\psi_{+,1}(z)}{h_{11}(z)}.
\end{align}
\end{subequations}

In addition, the modified auxiliary eigenfunctions are defined
\begin{subequations}\label{2.62}
\begin{align}
c_{1}(z)&=\mu_{1}(z)\mathrm{e}^{-\mathrm{i}\delta_{3}(z)}, \quad
c_{2}(z)=\mu_{2}(z)\mathrm{e}^{-\mathrm{i}\delta_{3}(z)},  \\
c_{3}(z)&=\mu_{3}(z)\mathrm{e}^{-\mathrm{i}\delta_{1}(z)}, \quad
c_{4}(z)=\mu_{4}(z)\mathrm{e}^{-\mathrm{i}\delta_{1}(z)},
\end{align}
\end{subequations}
the modified auxiliary eigenfunctions $c_{j}(z) (j=1,2,3,4)$ remain bounded for all $x\in\mathbb{R}$.

\subsection{Symmetries}
\label{s:Dsp:Sym}

In order to analyze the merging of a Riemann surface on each sheet, there are significant challenges in maintaining symmetries for NZBCs. For symmetry, there are transformations: $z\mapsto z^{*}$ and $z\mapsto -q_{0}^{2}/z$.

\subsubsection{First symmetry}
\label{s:Dsp:Sym:1}

Consider the transformation $z\mapsto z^{*}$, implying $(k,\lambda)\mapsto (k^{*},\lambda^{*})$.

\begin{proposition}\label{pro:2}
If $\psi(x,t,z)$ is a non-singular solution of the Lax pairs~\eqref{2.1}, so is
\begin{equation}\label{2.63}
\begin{split}
\mathbf{v}_{3}(x,t,z)=[\psi^{\dagger}(x,t,z^{*})]^{-1}.
\end{split}
\end{equation}
\end{proposition}

By using property $[\mathbf{e}^{\mathrm{i}\mathbf{\Omega}(z^{*})}]^{\dagger}
=\mathbf{e}^{-\mathrm{i}\mathbf{\Omega}(z)}$ and~\eqref{2.21}, the Jost eigenfunctions have the following symmetries
\begin{equation}\label{2.64}
\begin{split}
\psi_{\pm}(z)=[\psi_{\pm}^{\dagger}(z^{*})]^{-1}\mathbf{J}_{3}(z), \quad
\mathbf{J}_{3}(z)=\operatorname{diag}(\rho(z),1,\rho(z)), \quad z\in\Sigma.
\end{split}
\end{equation}

Then, using the Schwarz reflection principle,~\eqref{2.61} and~\eqref{2.64} yields the following
\begin{subequations}\label{2.66}
\begin{align}
\psi_{+,1}^{*}(z^{*})&=\frac{\mathrm{e}^{-\mathrm{i}[\delta_{1}+\delta_{2}
+\delta_{3}](z)}}{h_{22}(z)}[\psi_{+,2}\times\mu_{1}](z), \quad
\psi_{-,1}^{*}(z^{*})=\frac{\mathrm{e}^{-\mathrm{i}[\delta_{1}+\delta_{2}
+\delta_{3}](z)}}{s_{22}(z)}[\psi_{-,2}\times\mu_{2}](z),   \\
\psi_{+,3}^{*}(z^{*})&=\frac{\mathrm{e}^{-\mathrm{i}[\delta_{1}+\delta_{2}
+\delta_{3}](z)}}{h_{22}(z)}[\mu_{4}\times\psi_{+,2}](z), \quad
\psi_{-,3}^{*}(z^{*})=\frac{\mathrm{e}^{-\mathrm{i}[\delta_{1}+\delta_{2}
+\delta_{3}](z)}}{s_{22}(z)}[\mu_{3}\times\psi_{-,2}](z), \\
\psi_{+,2}^{*}(z^{*})&=\frac{\mathrm{e}^{-\mathrm{i}[\delta_{1}+\delta_{2}
+\delta_{3}](z)}}{\rho(z)h_{11}(z)}[\mu_{2}\times\psi_{+,1}](z)
=\frac{\mathrm{e}^{-\mathrm{i}[\delta_{1}+\delta_{2}
+\delta_{3}](z)}}{\rho(z)h_{33}(z)}[\psi_{+,3}\times\mu_{3}](z),  \\
\psi_{-,2}^{*}(z^{*})&=\frac{\mathrm{e}^{-\mathrm{i}[\delta_{1}+\delta_{2}
+\delta_{3}](z)}}{\rho(z)s_{11}(z)}[\mu_{1}\times\psi_{-,1}](z)
=\frac{\mathrm{e}^{-\mathrm{i}[\delta_{1}+\delta_{2}
+\delta_{3}](z)}}{\rho(z)s_{33}(z)}[\psi_{-,3}\times\mu_{4}](z).
\end{align}
\end{subequations}

According to~\eqref{2.39} and~\eqref{2.64}, the scattering matrices have the following relation
\begin{equation}\label{2.67}
\begin{split}
\mathbf{S}^{\dagger}(z^{*})=\mathbf{J}_{3}(z)\mathbf{H}(z)\mathbf{J}_{3}^{-1}(z), \quad z\in\Sigma.
\end{split}
\end{equation}
Correspondingly, we can find
\begin{subequations}\label{2.68}
\begin{align}
h_{11}(z)&=s_{11}^{*}(z^{*}), \quad \quad \quad
h_{12}(z)=\frac{s_{21}^{*}(z^{*})}{\rho(z)}, \quad  h_{13}(z)=s_{31}^{*}(z^{*}),  \\
h_{21}(z)&=\rho(z)s_{12}^{*}(z^{*}), \quad  h_{22}(z)=s_{22}^{*}(z^{*}), \quad
h_{23}(z)=\rho(z)s_{32}^{*}(z^{*}),  \\
h_{31}(z)&=s_{13}^{*}(z^{*}), \quad \quad \quad
h_{32}(z)=\frac{s_{23}^{*}(z^{*})}{\rho(z)}, \quad h_{33}(z)=s_{33}^{*}(z^{*}),
\end{align}
\end{subequations}
the analytical region of scattering coefficients on $z\in\Sigma$
\begin{subequations}\label{2.69}
\begin{align}
h_{11}(z)&=s_{11}^{*}(z^{*}), \quad  z\in\mathbb{D}_{2}, \\
h_{22}(z)&=s_{22}^{*}(z^{*}), \quad  \mathfrak{Im}z>0, \\
h_{33}(z)&=s_{33}^{*}(z^{*}), \quad  z\in\mathbb{D}_{3}.
\end{align}
\end{subequations}

Then, through the Schwarz reflection principle and new auxiliary eigenfunctions~\eqref{2.57}, we can derive symmetry relations as well
\begin{subequations}\label{2.70}
\begin{align}
\mu_{1}^{*}(z^{*})&=\mathrm{e}^{-\mathrm{i}[\delta_{1}+\delta_{2}+\delta_{3}](z)}
[\psi_{+,1}\times\psi_{-,2}](z), \quad  z\in\mathbb{U}_{2},  \\
\mu_{2}^{*}(z^{*})&=\mathrm{e}^{-\mathrm{i}[\delta_{1}+\delta_{2}+\delta_{3}](z)}
[\psi_{-,1}\times\psi_{+,2}](z), \quad  z\in\mathbb{U}_{1}, \\
\mu_{3}^{*}(z^{*})&=\mathrm{e}^{-\mathrm{i}[\delta_{1}+\delta_{2}+\delta_{3}](z)}
[\psi_{+,2}\times\psi_{-,3}](z), \quad  z\in\mathbb{U}_{4}, \\
\mu_{4}^{*}(z^{*})&=\mathrm{e}^{-\mathrm{i}[\delta_{1}+\delta_{2}+\delta_{3}](z)}
[\psi_{-,2}\times\psi_{+,3}](z), \quad  z\in\mathbb{U}_{3}.
\end{align}
\end{subequations}
In addition there are also symmetrical properties
\begin{equation}\label{2.71}
\begin{split}
\psi_{\pm,j}^{*}(z^{*})&=\frac{\mathrm{e}^{-\mathrm{i}[\delta_{1}+\delta_{2}
+\delta_{3}](z)}}{\rho_{j}(z)}[\psi_{\pm,l}\times\psi_{\pm,m}](z), \quad z\in\Sigma,
\end{split}
\end{equation}
where $j$, $l$ and $m$ are cyclic indices.

\subsubsection{Second symmetry}
\label{s:Dsp:Sym:2}

Consider the transformation $z\mapsto -q_{0}^{2}/z$, implying $(k,\lambda)\mapsto(k,-\lambda)$.

\begin{proposition}\label{pro:3}
If $\psi(z)$ is a non-singular solution of the Lax pairs~\eqref{2.1}, so is $\mathbf{v}_{4}(z)=\psi(-q_{0}^{2}/z)$.
\end{proposition}

The following properties can be obtained by progressiveness
\begin{equation}\label{2.72}
\begin{split}
\psi_{\pm}(z)=\psi_{\pm}(-\frac{q_{0}^{2}}{z})\mathbf{J}_{4}(z), \quad \mathbf{J}_{4}(z)=\begin{pmatrix}
    0 & 0 &  \displaystyle{-\frac{\mathrm{i}q_{0}}{z}}  \\
    0 & 1 &  0  \\
    \displaystyle{-\frac{\mathrm{i}q_{0}}{z}} & 0 & 0
  \end{pmatrix}, \quad z\in\Sigma.
\end{split}
\end{equation}

As before, the eigenfunctions have the following analytical properties
\begin{subequations}\label{2.74}
\begin{align}
\psi_{\pm,1}(z)&=-\frac{\mathrm{i}q_{0}}{z}\psi_{\pm,3}(-\frac{q_{0}^{2}}{z}), \quad
\{z\;|\; \mathfrak{Im}z\lessgtr0  \; \text{and} \; |z|>q_{0} \},  \\
\psi_{\pm,2}(z)&=\psi_{\pm,2}(-\frac{q_{0}^{2}}{z}), \quad \quad \quad
\{z\;|\; \operatorname{Im}z\gtrless0  \},  \\
\psi_{\pm,3}(z)&=-\frac{\mathrm{i}q_{0}}{z}\psi_{\pm,1}(-\frac{q_{0}^{2}}{z}), \quad
\{z\;|\; \mathfrak{Im}z\lessgtr0  \; \text{and} \; |z|<q_{0}  \},
\end{align}
\end{subequations}
the scattering matrices have the following relation
\begin{equation}\label{2.75}
\begin{split}
\mathbf{S}(-\frac{q_{0}^{2}}{z})=\mathbf{J}_{4}(z)\mathbf{S}(z)\mathbf{J}_{4}^{-1}(z), \quad z\in\Sigma,
\end{split}
\end{equation}
Componentwise, we have
\begin{subequations}\label{2.76}
\begin{align}
s_{11}(z)&=s_{33}(-\frac{q_{0}^{2}}{z}), \quad \quad \quad s_{12}(z)=\frac{\mathrm{i}z}{q_{0}}s_{32}(-\frac{q_{0}^{2}}{z}), \quad
s_{13}(z)=s_{31}(-\frac{q_{0}^{2}}{z}),  \\
s_{21}(z)&=-\frac{\mathrm{i}q_{0}}{z}s_{23}(-\frac{q_{0}^{2}}{z}), \quad  s_{22}(z)=s_{22}(-\frac{q_{0}^{2}}{z}), \quad \quad
s_{23}(z)=-\frac{\mathrm{i}q_{0}}{z}s_{21}(-\frac{q_{0}^{2}}{z}),  \\
s_{31}(z)&=s_{13}(-\frac{q_{0}^{2}}{z}), \quad \quad \quad
s_{32}(z)=\frac{\mathrm{i}z}{q_{0}}s_{12}(-\frac{q_{0}^{2}}{z}), \quad
s_{33}(z)=s_{11}(-\frac{q_{0}^{2}}{z}),
\end{align}
\end{subequations}
the analytical region of scattering coefficients on $z\in\Sigma$
\begin{subequations}\label{2.77}
\begin{align}
s_{11}(z)&=s_{33}(-\frac{q_{0}^{2}}{z}), \quad z\in\mathbb{U}_{1}, \quad \quad
s_{22}(z)=s_{22}(-\frac{q_{0}^{2}}{z}), \quad \mathfrak{Im}z<0, \\
h_{11}(z)&=h_{33}(-\frac{q_{0}^{2}}{z}), \quad z\in\mathbb{U}_{2}, \quad \quad
h_{22}(z)=h_{22}(-\frac{q_{0}^{2}}{z}), \quad \mathfrak{Im}z>0.
\end{align}
\end{subequations}

The auxiliary eigenfunctions satisfy the following properties
\begin{equation}\label{2.78}
\begin{split}
\mu_{1}(z)&=-\frac{\mathrm{i}q_{0}}{z}\mu_{4}(-\frac{q_{0}^{2}}{z}), \quad
(z\in\mathbb{U}_{1}),  \quad
\mu_{2}(z)=-\frac{\mathrm{i}q_{0}}{z}\mu_{3}(-\frac{q_{0}^{2}}{z}), \quad
(z\in\mathbb{U}_{2}).
\end{split}
\end{equation}
In light of the first and second symmetry connections governing the associated scattering coefficients, we bring forth novel reflections as follows
\begin{subequations}\label{2.79}
\begin{align}
\alpha_{1}(z)&=\frac{s_{21}(z)}{s_{11}(z)}=\rho(z)\frac{h_{12}^{*}(z^{*})}{h_{11}^{*}(z^{*})}, \quad
\alpha_{1}(-\frac{q_{0}^{2}}{z})=\frac{\mathrm{i}z}{q_{0}}\frac{s_{23}(z)}{s_{33}(z)}
=\frac{\mathrm{i}z\rho(z)}{q_{0}}\frac{h_{32}^{*}(z^{*})}{h_{33}^{*}(z^{*})}, \\
\alpha_{2}(z)&=\frac{s_{31}(z)}{s_{11}(z)}=\frac{h_{13}^{*}(z^{*})}{h_{11}^{*}(z^{*})}, \quad \quad \quad \alpha_{2}(-\frac{q_{0}^{2}}{z})=\frac{s_{13}(z)}{s_{33}(z)}
=\frac{h_{31}^{*}(z^{*})}{h_{33}^{*}(z^{*})}, \\
\alpha_{3}(z)&=\frac{s_{32}(z)}{s_{22}(z)}=\frac{1}{\rho(z)}\frac{h_{23}^{*}(z^{*})}{h_{22}^{*}(z^{*})}, \quad \alpha_{3}(-\frac{q_{0}^{2}}{z})=-\frac{\mathrm{i}q_{0}}{z}\frac{s_{12}(z)}{s_{22}(z)}
=-\frac{\mathrm{i}q_{0}}{z\rho(z)}\frac{h_{21}^{*}(z^{*})}{h_{22}^{*}(z^{*})}.
\end{align}
\end{subequations}

\section{Discrete spectrum and its related properties}
\label{s:Dsirp}

\subsection{Discrete spectrum}
\label{s:Dsirp:Ds}

To analyze the discrete spectrum, the following four matrices are defined
\begin{subequations}\label{3.1}
\begin{align}
\mathbf{\Psi}_{1}(z)&=(\psi_{-,1}(z),\psi_{+,2}(z),\mu_{1}(z)), \quad z\in\mathbb{U}_{1}, \\
\mathbf{\Psi}_{2}(z)&=(\psi_{+,1}(z),\psi_{-,2}(z),\mu_{2}(z)), \quad z\in\mathbb{U}_{2}, \\
\mathbf{\Psi}_{3}(z)&=(\mu_{3}(z),\psi_{-,2}(z),\psi_{+,3}(z)), \quad z\in\mathbb{U}_{3}, \\
\mathbf{\Psi}_{4}(z)&=(\mu_{4}(z),\psi_{+,2}(z),\psi_{-,3}(z)), \quad z\in\mathbb{U}_{4}.
\end{align}
\end{subequations}

By the~\eqref{2.61}, it yields
\begin{subequations}\label{3.2}
\begin{align}
\operatorname{det}\mathbf{\Psi}_{1}(z)&=s_{11}(z)h_{22}(z)\rho(z)
\mathrm{e}^{\mathrm{i}[\delta_{1}+\delta_{2}+\delta_{3}](z)}, \quad z\in\mathbb{U}_{1}, \\
\operatorname{det}\mathbf{\Psi}_{2}(z)&=s_{22}(z)h_{11}(z)\rho(z)
\mathrm{e}^{\mathrm{i}[\delta_{1}+\delta_{2}+\delta_{3}](z)}, \quad z\in\mathbb{U}_{2}, \\
\operatorname{det}\mathbf{\Psi}_{3}(z)&=s_{22}(z)h_{33}(z)\rho(z)
\mathrm{e}^{\mathrm{i}[\delta_{1}+\delta_{2}+\delta_{3}](z)}, \quad z\in\mathbb{U}_{3}, \\
\operatorname{det}\mathbf{\Psi}_{4}(z)&=s_{33}(z)h_{22}(z)\rho(z)
\mathrm{e}^{\mathrm{i}[\delta_{1}+\delta_{2}+\delta_{3}](z)}, \quad z\in\mathbb{U}_{4}.
\end{align}
\end{subequations}

Thus, the columns of $\mathbf{\Psi}_{1}(z)$ become linearly dependent at the zeros of $s_{11}(z)$, $h_{22}(z)$ and $\rho(z)$, Similarly, columns $\mathbf{\Psi}_{2}(z)$, $\mathbf{\Psi}_{3}(z)$ and $\mathbf{\Psi}_{4}(z)$ have the same properties. However, the symmetries of the scattering coefficients suggest that these zeros are interconnected and not independent of one another.

\begin{proposition}\label{pro:4}
Let $\mathfrak{Im}z>0$, then
\begin{equation}\label{3.3}
\begin{split}
h_{22}(z)=0 \Longleftrightarrow s_{22}(z^{*})=0 \Longleftrightarrow s_{22}(-\frac{q_{0}^{2}}{z^{*}})=0
\Longleftrightarrow h_{22}(-\frac{q_{0}^{2}}{z})=0.
\end{split}
\end{equation}
\end{proposition}

\begin{proposition}\label{pro:5}
Let $\mathfrak{Im}z>0$ and $|z|>q_{0}$, then
\begin{equation}\label{3.4}
\begin{split}
s_{11}(z)=0 \Longleftrightarrow h_{11}(z^{*})=0 \Longleftrightarrow h_{33}(-\frac{q_{0}^{2}}{z^{*}})=0
\Longleftrightarrow s_{33}(-\frac{q_{0}^{2}}{z})=0.
\end{split}
\end{equation}
\end{proposition}

Propositions~\ref{pro:4} and~\ref{pro:5} suggest that symmetric quartets give rise to discrete eigenvalues: $z$, $z^{*}$, $-q_{0}^{2}/z$ and $-q_{0}^{2}/z^{*}$. Hence, there are three potential types of eigenvalues related to $z\in\mathbb{U}_{1}$:
\begin{enumerate}
  \item $s_{11}(z)=0$ and $h_{22}(z)\neq0$, we call it the first type of eigenvalues.
  \item $s_{11}(z)\neq0$ and $h_{22}(z)=0$, we call it the second type of eigenvalues.
  \item $s_{11}(z)=h_{22}(z)=0$, we call it the third type of eigenvalues.
\end{enumerate}

We proceed to describe these three different types of eigenvalues, along with their respective properties at these discrete points.

\begin{proposition}\label{pro:6}
Suppose $\mathfrak{Im}z>0$ and $|z|>q_{0}$, the following conclusions are equivalent
\begin{itemize} \itemsep0.75pt
  \item $\mu_{1}(z)=\mathbf{0}$.
  \item $\mu_{4}(-\displaystyle{\frac{q_{0}^{2}}{z}})=\mathbf{0}$.
  \item \textit{$\psi_{-,2}(z^{*})$ and $\psi_{+,1}(z^{*})$ are linearly correlated}.
  \item \textit{$\psi_{+,3}(-\displaystyle{\frac{q_{0}^{2}}{z^{*}}})$ and $\psi_{-,2}(-\displaystyle{\frac{q_{0}^{2}}{z^{*}}})$ are linearly correlated}.
\end{itemize}
\end{proposition}

\begin{proposition}\label{pro:7}
Assume that $\mathfrak{Im}z>0$ and $|z|>q_{0}$, the following conclusions are equivalent
\begin{itemize} \itemsep0.75pt
  \item $\mu_{2}(z^{*})=\mathbf{0}$.
  \item $\mu_{3}(-\displaystyle{\frac{q_{0}^{2}}{z^{*}}})=\mathbf{0}$.
  \item \textit{$\psi_{-,1}(z)$ and $\psi_{+,2}(z)$ are linearly correlated}.
  \item \textit{$\psi_{+,2}(-\displaystyle{\frac{q_{0}^{2}}{z}})$ and $\psi_{-,3}(-\displaystyle{\frac{q_{0}^{2}}{z}})$ are linearly correlated}.
\end{itemize}
\end{proposition}

Assuming that the discrete eigenvalues are simple, we can conclude that

\begin{proposition}\label{pro:8}
If $z_{m}$ is the first type eigenvalue with $z_{m}\in\mathbb{U}_{1}$, there exist constants $d_{m}$, $\hat{d}_{m}$, $\check{d}_{m}$ and $\bar{d}_{m}$ such that
\begin{equation}\label{3.5}
\begin{split}
\psi_{-,1}(z_{m})&=\frac{d_{m}}{h_{22}(z_{m})}\mu_{1}(z_{m}), \quad
\mu_{2}(z_{m}^{*})=\hat{d}_{m}\psi_{+,1}(z_{m}^{*}), \\
\mu_{3}(-\frac{q_{0}^{2}}{z_{m}^{*}})&=\check{d}_{m}\psi_{+,3}(-\frac{q_{0}^{2}}{z_{m}^{*}}), \quad
\psi_{-,3}(-\frac{q_{0}^{2}}{z_{m}})=\bar{d}_{m}\mu_{4}(-\frac{q_{0}^{2}}{z_{m}}).
\end{split}
\end{equation}
\end{proposition}

\begin{proposition}\label{pro:9}
If $\theta_{m}$ is the second type eigenvalue with $\theta_{m}\in\mathbb{U}_{1}$, there exist constants $g_{m}$, $\hat{g}_{m}$, $\check{g}_{m}$ and $\bar{g}_{m}$ such that
\begin{equation}\label{3.6}
\begin{split}
\mu_{1}(\theta_{m})&=g_{m}\psi_{+,2}(\theta_{m}), \quad
\psi_{-,2}(\theta_{m}^{*})=\hat{g}_{m}\mu_{2}(\theta_{m}^{*}), \\
\psi_{-,2}(-\frac{q_{0}^{2}}{\theta_{m}^{*}})&=\check{g}_{m}\mu_{3}(-\frac{q_{0}^{2}}{\theta_{m}^{*}}), \quad
\mu_{4}(-\frac{q_{0}^{2}}{\theta_{m}})=\bar{g}_{m}\psi_{+,2}(-\frac{q_{0}^{2}}{\theta_{m}}).
\end{split}
\end{equation}
\end{proposition}

\begin{proposition}\label{pro:10}
If $\xi_{m}$ is the third type eigenvalue with $\xi_{m}\in\mathbb{U}_{1}$ and $\mu_{1}(\xi_{m})=\mu_{2}(\xi_{m}^{*})=\mathbf{0}$, there exist constants $f_{m}$, $\hat{f}_{m}$, $\check{f}_{m}$ and $\bar{f}_{m}$ such that
\begin{equation}\label{3.7}
\begin{split}
\psi_{-,1}(\xi_{m})&=f_{m}\psi_{+,2}(\xi_{m}), \quad \quad
\psi_{-,2}(\xi_{m}^{*})=\hat{f}_{m}\psi_{+,1}(\xi_{m}^{*}), \\
\psi_{-,2}(-\frac{q_{0}^{2}}{\xi_{m}^{*}})&=\check{f}_{m}\psi_{+,3}(-\frac{q_{0}^{2}}{\xi_{m}^{*}}), \quad
\psi_{-,3}(-\frac{q_{0}^{2}}{\xi_{m}})=\bar{f}_{m}\psi_{+,2}(-\frac{q_{0}^{2}}{\xi_{m}}).
\end{split}
\end{equation}
\end{proposition}

The combination of Propositions~\ref{pro:8},~\ref{pro:9} and~\ref{pro:10} offers a comprehensive understanding of the discrete spectrum. Specifically, it becomes evident that discrete eigenvalues of each type correspond to bound states in the scattering problem.

Let $\{z_{m}\}_{m=1}^{\mathrm{M}_{1}}$ represent all values of the first type of eigenvalues, then
\begin{subequations}\label{3.8}
\begin{align}
\nu_{-,1}(z_{m})&=\frac{d_{m}}{h_{22}(z_{m})}c_{1}(z_{m})
\mathrm{e}^{\mathrm{i}[\delta_{3}-\delta_{1}](z_{m})},  \quad
c_{2}(z_{m}^{*})=\hat{d}_{m}\nu_{+,1}(z_{m}^{*})
\mathrm{e}^{\mathrm{i}[\delta_{1}-\delta_{3}](z_{m}^{*})}, \\
c_{3}(-\frac{q_{0}^{2}}{z_{m}^{*}})&=\check{d}_{m}\nu_{+,3}(-\frac{q_{0}^{2}}{z_{m}^{*}})
\mathrm{e}^{\mathrm{i}[\delta_{1}+\delta_{3}](z_{m}^{*})}, \quad
\nu_{-,3}(-\frac{q_{0}^{2}}{z_{m}})=\bar{d}_{m}c_{4}(-\frac{q_{0}^{2}}{z_{m}})
\mathrm{e}^{\mathrm{i}[\delta_{3}-\delta_{1}](z_{m})}.
\end{align}
\end{subequations}

Let $\{\theta_{m}\}_{m=1}^{\mathrm{m}_{2}}$ represent all values of the second type of eigenvalues, then
\begin{subequations}\label{3.9}
\begin{align}
c_{1}(\theta_{m})&=g_{m}\nu_{+,2}(\theta_{m})
\mathrm{e}^{\mathrm{i}[\delta_{2}-\delta_{3}](\theta_{m})}, \quad
\nu_{-,2}(\theta_{m}^{*})=\hat{g}_{m}c_{2}(\theta_{m}^{*})
\mathrm{e}^{\mathrm{i}[\delta_{3}-\delta_{2}](\theta_{m}^{*})}, \\
\nu_{-,2}(-\frac{q_{0}^{2}}{\theta_{m}^{*}})&=\check{g}_{m}c_{3}(-\frac{q_{0}^{2}}{\theta_{m}^{*}})
\mathrm{e}^{\mathrm{i}[\delta_{3}-\delta_{2}](\theta_{m}^{*})}, \quad
c_{4}(-\frac{q_{0}^{2}}{\theta_{m}})=\bar{g}_{m}\nu_{+,2}(-\frac{q_{0}^{2}}{\theta_{m}})
\mathrm{e}^{\mathrm{i}[\delta_{2}-\delta_{3}](\theta_{m})}.
\end{align}
\end{subequations}

Let $\{\xi_{m}\}_{m=1}^{\mathrm{M}_{3}}$ represent all values of the third type of eigenvalues, then
\begin{subequations}\label{3.10}
\begin{align}
\nu_{-,1}(\xi_{m})&=f_{m}\nu_{+,2}(\xi_{m})
\mathrm{e}^{\mathrm{i}[\delta_{2}-\delta_{1}](\xi_{m})}, \quad
\nu_{-,2}(\xi_{m}^{*})=\hat{f}_{m}\nu_{+,1}(\xi_{m}^{*})
\mathrm{e}^{\mathrm{i}[\delta_{1}-\delta_{2}](\xi_{m}^{*})}, \\
\nu_{-,2}(-\frac{q_{0}^{2}}{\xi_{m}^{*}})&=
\check{f}_{m}\nu_{+,3}(-\frac{q_{0}^{2}}{\xi_{m}^{*}})
\mathrm{e}^{\mathrm{i}[\delta_{1}-\delta_{2}](\xi_{m}^{*})}, \quad
\nu_{-,3}(-\frac{q_{0}^{2}}{\xi_{m}})=\bar{f}_{m}\nu_{-,2}(-\frac{q_{0}^{2}}{\xi_{m}})
\mathrm{e}^{\mathrm{i}[\delta_{2}-\delta_{1}](\xi_{m})}.
\end{align}
\end{subequations}

\begin{proposition}\label{pro:11}
The norming constants in Propositions~\ref{pro:8},~\ref{pro:9} and~\ref{pro:10} obey the symmetry relations
\begin{subequations}\label{3.11}
\begin{align}
\bar{g}_{m}&=\frac{\mathrm{i}\theta_{m}}{q_{0}}g_{m}, \quad \hat{g}_{m}=-\frac{g_{m}^{*}}{\rho(\theta_{m}^{*})h_{11}(\theta_{m}^{*})}, \quad
\check{g}_{m}=\frac{\mathrm{i}q_{0}}{\theta_{m}^{*}}
\frac{g_{m}^{*}}{\rho(\theta_{m}^{*})h_{11}(\theta_{m}^{*})}, \\
\bar{f}_{m}&=\frac{\mathrm{i}\xi_{m}}{q_{0}}f_{m}, \quad
\hat{f}_{m}=-\frac{s_{22}^{\prime}(\xi_{m}^{*})}{h_{11}^{\prime}(\xi_{m}^{*})}
\frac{f_{m}^{*}}{\rho(\xi_{m}^{*})}, \quad \check{f}_{m}=\frac{\mathrm{i}q_{0}}{\xi_{m}^{*}}
\frac{s_{22}^{\prime}(\xi_{m}^{*})}{h_{11}^{\prime}(\xi_{m}^{*})}\frac{f_{m}^{*}}{\rho(\xi_{m}^{*})}, \\
\bar{d}_{m}&=\frac{d_{m}}{h_{22}(z_{m})}, \quad \hat{d}_{m}=\check{d}_{m}=-d_{m}^{*}.
\end{align}
\end{subequations}
\end{proposition}

\subsection{Asymptotic behavior}
\label{s:Dsirp:Ab}

With the help of the Wentzel-Kramers-Brillouin expansion method~\cite{B6}, the asymptotic behavior of the modified eigenfunctions can be analyzed. Specifically, combined with the differential equations~\eqref{2.26}, so it has asymptotic properties
\begin{subequations}
\begin{align}
\nu_{\pm,1}(x,t,z)&=\begin{pmatrix}
    1   \\
    -\mathrm{i}\mathbf{q}/z
  \end{pmatrix}+O(z^{-2}), \quad z\rightarrow\infty, \label{3.13a} \\
\nu_{\pm,2}(x,t,z)&=\begin{pmatrix}
    -\mathrm{i}\mathbf{q}^{\dagger}\mathbf{q}_{\pm}^{\perp}/q_{0}z   \\
    \mathbf{q}_{\pm}^{\perp}/q_{0}
  \end{pmatrix}+O(z^{-2}), \quad z\rightarrow\infty, \label{3.13b} \\
\nu_{\pm,3}(x,t,z)&=\begin{pmatrix}
    -\mathrm{i}\mathbf{q}^{\dagger}\mathbf{q}_{\pm}/q_{0}z   \\
    \mathbf{q}_{\pm}/q_{0}
  \end{pmatrix}+O(z^{-2}), \quad z\rightarrow\infty, \label{3.13c}
\end{align}
\end{subequations}
and
\begin{subequations}\label{3.14}
\begin{align}
\nu_{\pm,1}(x,t,z)&=\begin{pmatrix}
    \mathbf{q}^{\dagger}\mathbf{q}_{\pm}/q_{0}^{2}  \\
    -\mathrm{i}\mathbf{q}_{\pm}/z
  \end{pmatrix}+O(z), \quad z\rightarrow0, \\
\nu_{\pm,2}(x,t,z)&=\begin{pmatrix}
    0   \\
    \mathbf{q}_{\pm}^{\perp}/q_{0}
  \end{pmatrix}+O(z), \quad z\rightarrow0, \\
\nu_{\pm,3}(x,t,z)&=\begin{pmatrix}
    -\mathrm{i}q_{0}/z   \\
    \mathbf{q}/q_{0}
  \end{pmatrix}+O(z), \quad z\rightarrow0.
\end{align}
\end{subequations}

For $j=1,2,3,4$ in terms of definition of the new auxiliary eigenfunctions $\mu_{j}(z)$ and the modified auxiliary eigenfunctions $c_{j}(z)$, combining~\eqref{2.57} with~\eqref{2.70} we can obtain
\begin{equation}\begin{split}
c_{1}(z)&=\begin{pmatrix}
    -\mathrm{i}\mathbf{q}^{\dagger}\mathbf{q}_{-}/q_{0}z  \\
    \mathbf{q}_{-}/q_{0}
  \end{pmatrix}+O(z^{-2}), \quad
c_{2}(z)=\begin{pmatrix}
    -\mathrm{i}\mathbf{q}^{\dagger}\mathbf{q}_{+}/q_{0}z  \\
    \mathbf{q}_{+}/q_{0}
  \end{pmatrix}+O(z^{-2}), \quad z\rightarrow\infty, \\
c_{3}(z)&=\begin{pmatrix}
    \mathbf{q}_{-}^{\dagger}\mathbf{q}_{+}/q_{0}^{2}     \\
    -\mathrm{i}(\mathbf{q}_{-}^{\dagger}\mathbf{q}_{+})\mathbf{q}/q_{0}^{2}z
  \end{pmatrix}+O(z^{-2}), \quad
c_{4}(z)=\begin{pmatrix}
    \mathbf{q}_{+}^{\dagger}\mathbf{q}_{-}/q_{0}^{2}     \\
    -\mathrm{i}(\mathbf{q}_{+}^{\dagger}\mathbf{q}_{-})\mathbf{q}/q_{0}^{2}z
  \end{pmatrix}+O(z^{-2}), \quad z\rightarrow\infty,
\end{split}\end{equation}
and
\begin{equation}\begin{split}
c_{1}(z)&=\begin{pmatrix}
    -\mathrm{i}\mathbf{q}_{+}^{\dagger}\mathbf{q}_{-}/q_{0}z  \\
    (\mathbf{q}_{+}^{\dagger}\mathbf{q})\mathbf{q}_{-}/q_{0}^{3}
  \end{pmatrix}+O(z), \quad
c_{2}(z)=\begin{pmatrix}
    -\mathrm{i}\mathbf{q}_{-}^{\dagger}\mathbf{q}_{+}/q_{0}z  \\
    (\mathbf{q}_{-}^{\dagger}\mathbf{q})\mathbf{q}_{+}/q_{0}^{3}
  \end{pmatrix}+O(z), \quad  z\rightarrow0, \\
c_{3}(z)&=\begin{pmatrix}
    \mathbf{q}^{\dagger}\mathbf{q}_{+}/q_{0}^{2}   \\
    -\mathrm{i}\mathbf{q}_{+}/z
  \end{pmatrix}+O(z), \quad
c_{4}(z)=\begin{pmatrix}
    \mathbf{q}^{\dagger}\mathbf{q}_{-}/q_{0}^{2}   \\
    -\mathrm{i}\mathbf{q}_{-}/z
  \end{pmatrix}+O(z), \quad  z\rightarrow0.
\end{split}\end{equation}

\begin{corollary}\label{cor:2}
As $z\rightarrow\infty$ in the appropriate regions of the complex plane
\begin{subequations}\label{3.15}
\begin{align}
s_{11}(z)&=1+O(\frac{1}{z}), \quad
s_{22}(z)=\frac{\mathbf{q}_{-}^{\dagger}\mathbf{q}_{+}}{q_{0}^{2}}+O(\frac{1}{z}), \quad
s_{33}(z)=\frac{\mathbf{q}_{+}^{\dagger}\mathbf{q}_{-}}{q_{0}^{2}}+O(\frac{1}{z}), \\
h_{11}(z)&=1+O(\frac{1}{z}), \quad
h_{22}(z)=\frac{\mathbf{q}_{+}^{\dagger}\mathbf{q}_{-}}{q_{0}^{2}}+O(\frac{1}{z}), \quad
h_{33}(z)=\frac{\mathbf{q}_{-}^{\dagger}\mathbf{q}_{+}}{q_{0}^{2}}+O(\frac{1}{z}).
\end{align}
\end{subequations}
Similarly, as $z\rightarrow0$ in the appropriate regions of the complex plane
\begin{subequations}\label{3.16}
\begin{align}
s_{11}(\mathbf{z})&=(\mathbf{q}_{+}^{\dagger}\mathbf{q}_{-})/q_{0}^{2}+O(z), \quad
s_{22}(\mathbf{z})=(\mathbf{q}_{-}^{\dagger}\mathbf{q}_{+})/q_{0}^{2}+O(z), \quad
s_{33}(\mathbf{z})=1+O(z),  \\
h_{11}(\mathbf{z})&=(\mathbf{q}_{-}^{\dagger}\mathbf{q}_{+})/q_{0}^{2}+O(z), \quad
h_{22}(\mathbf{z})=(\mathbf{q}_{+}^{\dagger}\mathbf{q}_{-})/q_{0}^{2}+O(z), \quad
h_{33}(\mathbf{z})=1+O(z).
\end{align}
\end{subequations}
\end{corollary}

The asymptotic behavior of other scattering coefficients
\begin{equation}\label{3.17}
\begin{split}
[\mathbf{S}^{\pm1}(z)]_{o}=\frac{1}{q_{0}^{2}}\begin{pmatrix}
    0 & 0 & 0  \\
    0 & 0 & (\mathbf{q}_{\pm}^{\perp})^{\dagger}\mathbf{q}_{\mp}  \\
    0 & \mathbf{q}_{\pm}^{\dagger}\mathbf{q}_{\mp}^{\perp} & 0
  \end{pmatrix}+O(z^{-1}), \quad z\rightarrow\infty,
\end{split}
\end{equation}
and
\begin{equation}\label{3.18}
\begin{split}
[\mathbf{S}^{\pm1}(z)]_{o}=\frac{\mathrm{i}}{zq_{0}}\begin{pmatrix}
    0 & 0 & 0  \\
    (\mathbf{q}_{\mp}^{\perp})^{\dagger}\mathbf{q}_{\pm} & 0 & 0  \\
    0 & 0 & 0
  \end{pmatrix}+O(1), \quad z\rightarrow0.
\end{split}
\end{equation}

\subsection{Behavior at the branch points}
\label{s:Dsirp:Behavior}

We explore the characteristics of the scattering matrix at the branch points $k=\pm\mathrm{i}q_{0}$.
\begin{equation}
\begin{split}
\lim_{z\rightarrow\pm\mathrm{i}q_{0}}{\mathbf{Y}_{\pm}
\mathbf{e}^{\mathrm{i}(x-r)\mathbf{\Lambda}_{1}}\mathbf{Y}_{\pm}^{-1}}=\begin{pmatrix}
    1\pm q_{0}(x-r) & \mathbf{q}_{\pm}^{\dagger}(x-r)   \\
    -\mathbf{q}_{\pm}(x-r) & \displaystyle{\frac{1}{q_{0}^{2}}}(1\mp q_{0}(x-r))\mathbf{q}_{\pm}\mathbf{q}_{\pm}^{\dagger}
+\mathrm{e}^{\mp q_{0}(x-r)}\mathbf{q}_{\pm}^{\perp}(\mathbf{q}_{\pm}^{\perp})^{\dagger}   \\
  \end{pmatrix}.
\end{split}
\end{equation}

If $(1+|x|)\left(\mathbf{q}-\mathbf{q}_{ \pm}\right) \in L^{1}\left(\mathbb{R}^{\pm}\right)$, the integrals in~\eqref{2.29} are also convergent at $z\rightarrow\pm\mathrm{i}q_{0}$. Then, the columns of $\psi_{\pm}(\mathrm{i}q_{0})$ and $\psi_{\pm}(-\mathrm{i}q_{0})$ are linearly dependent
\begin{equation}\label{3.20}
\begin{split}
\psi_{\pm,1}(\mathrm{i}q_{0})=-\psi_{\pm,3}(\mathrm{i}q_{0}), \quad
\psi_{\pm,1}(-\mathrm{i}q_{0})=\psi_{\pm,3}(-\mathrm{i}q_{0}).
\end{split}
\end{equation}
We analyze the behavior of the scattering matrix $\mathbf{S}(z)$ near the branch points, which can be represented by Wronskians. For this reason, we record the scattering coefficients as Wronskians
\begin{equation}\label{3.21}
\begin{split}
s_{jl}(z)=\frac{z^{2}}{z^{2}+q_{0}^{2}}W_{jl}(z)
\mathrm{e}^{-\mathrm{i}[\delta_{1}+\delta_{2}+\delta_{3}](z)},
\end{split}
\end{equation}
where
\begin{equation}\label{3.22}
\begin{split}
W_{jl}(z)=\operatorname{det}(\psi_{-,l}(z),\psi_{+,j+1}(z),\psi_{+,j+2}(z)),
\end{split}
\end{equation}
with $j+1$ and $j+2$ are calculated modulo 3.

The Laurent series expansions of the scattering coefficients at $z\rightarrow\pm\mathrm{i}q_{0}$ are expressed as
\begin{equation}\label{3.23}
\begin{split}
s_{ij}(z)&=\frac{s_{ij,\pm}}{z\mp\mathrm{i}q_{0}}+s_{ij,\pm}^{(o)}+O(z\mp\mathrm{i}q_{0}), \quad z\in\Sigma \backslash \{\pm\mathrm{i}q_{0}\},
\end{split}
\end{equation}
where
\begin{equation}\label{3.24}
\begin{split}
s_{ij,\pm}&=\pm\frac{\mathrm{i}q_{0}}{2}W_{ij}(\pm\mathrm{i}q_{0})
\exp\left[ \pm q_{0}x+5\mathrm{i}q_{0}^{2}(1-\sigma q_{0}^{2})t \right], \\
s_{ij,\pm}^{(o)}&=\left[\pm\frac{\mathrm{i}q_{0}}{2}
\frac{\partial}{\partial z}W_{ij}(z)|_{z=\pm\mathrm{i}q_{0}}
+W_{ij}(\pm\mathrm{i}q_{0}) \right]\exp\left[ \pm q_{0}x+5\mathrm{i}q_{0}^{2}(1-\sigma q_{0}^{2})t \right].
\end{split}
\end{equation}
Then, the asymptotic expansion of $\mathbf{S}(z)$ around the branch point can be expressed as follows:
\begin{equation}\label{3.25}
\begin{split}
\mathbf{S}(z)&=\frac{\mathbf{S}_{\pm}}{z\mp\mathrm{i}q_{0}}+\mathbf{S}_{\pm}^{(o)}
+O(z\mp\mathrm{i}q_{0}),
\end{split}
\end{equation}
where $\mathbf{S}_{\pm}^{(o)}=(s_{ij,\pm}^{(o)})$,
\begin{equation}\label{3.26}
\begin{split}
\mathbf{S}_{\pm}=s_{11,\pm}\begin{pmatrix}
    1 & 0 & \mp1  \\
    0 & 0 & 0  \\
    \pm1 & 0 & -1
  \end{pmatrix}+s_{12,\pm}\begin{pmatrix}
    0 & 1 & 0  \\
    0 & 0 & 0  \\
    0 & \pm1 & 0
  \end{pmatrix}.
\end{split}
\end{equation}

\section{Inverse problem}
\label{s:Inverse problem}

\subsection{Riemann-Hilbert problem}
\label{s:Inverse problem:RHP}

To construct the matrix RH problem, we require appropriate jump conditions that articulate eigenfunctions, which are meromorphic in domain $\mathbb{U}_{j}(j=1,2,3,4)$. Then, the following Lemma can be given by using the desired eigenfunctions as the column of $\Psi_{j}(z)(j=1,2,3,4)$ in~\eqref{3.1} and the scattering relation~\eqref{2.39}.

\begin{lemma}\label{lem:1}
Define the piecewise meromorphic function $\mathbf{B}=\mathbf{B}(x,t,z)=\mathbf{B}_{j}(x,t,z)$ for $z\in\mathbb{U}_{j} (j=1,\cdots,4)$, where
\begin{subequations}\label{4.1}
\begin{align}
\mathbf{B}_{1}(x,t,z)&=\mathbf{\Psi}_{1}\mathbf{e}^{-\mathrm{i}\mathbf{\Omega}}
\left[\operatorname{diag}(\frac{1}{s_{11}},1,\frac{1}{h_{22}})\right]=
\left[\displaystyle{\frac{\nu_{-,1}}{s_{11}}}, \nu_{+,2}, \displaystyle{\frac{c_{1}}{h_{22}}}\right],
\quad  z\in\mathbb{U}_{1}, \\
\mathbf{B}_{2}(x,t,z)&=\mathbf{\Psi}_{2}\mathbf{e}^{-\mathrm{i}\mathbf{\Omega}}
\left[\operatorname{diag}(1,\frac{1}{s_{22}},\frac{1}{h_{11}})\right]=
\left[\nu_{+,1}, \displaystyle{\frac{\nu_{-,2}}{s_{22}}}, \displaystyle{\frac{c_{2}}{h_{11}}}\right],
\quad  z\in\mathbb{U}_{2}, \\
\mathbf{B}_{3}(x,t,z)&=\mathbf{\Psi}_{3}\mathbf{e}^{-\mathrm{i}\mathbf{\Omega}}
\left[\operatorname{diag}(\frac{1}{h_{33}},\frac{1}{s_{22}},1)\right]=
\left[\displaystyle{\frac{c_{3}}{h_{33}}}, \displaystyle{\frac{\nu_{-,2}}{s_{22}}}, \nu_{+,3}\right],  \quad  z\in\mathbb{U}_{3}, \\
\mathbf{B}_{4}(x,t,z)&=\mathbf{\Psi}_{4}\mathbf{e}^{-\mathrm{i}\mathbf{\Omega}}
\left[\operatorname{diag}(\frac{1}{h_{22}},1,\frac{1}{s_{33}})\right]=
\left[ \displaystyle{\frac{c_{4}}{h_{22}}}, \nu_{+,2}, \displaystyle{\frac{\nu_{-,3}}{s_{33}}}\right],
\quad  z\in\mathbb{U}_{4}.
\end{align}
\end{subequations}
Then the jump conditions can be obtained according to the reflection coefficients
\begin{equation}\label{4.2}
\begin{split}
\mathbf{B}^{+}(x,t,z)&=\mathbf{B}^{-}(x,t,z)[\mathbf{I}-\mathbf{e}^{\mathrm{i}\mathbf{\Omega}}
\mathbf{L}(z)\mathbf{e}^{-\mathrm{i}\mathbf{\Omega}}], \quad z\in\Sigma,
\end{split}
\end{equation}
where $\mathbf{B}=\mathbf{B}^{+}$ for $z\in\mathbb{U}^{+}=\mathbb{U}_{1}\cup\mathbb{U}_{3}$ and $\mathbf{B}=\mathbf{B}^{-}$ for $z\in\mathbb{U}^{-}=\mathbb{U}_{2}\cup\mathbb{U}_{4}$ $($with $\mathbf{B}^{+}=\mathbf{B}_{1}$ for $z\in\mathbb{U}_{1}$, $\mathbf{B}^{+}=\mathbf{B}_{3}$ for $z\in\mathbb{U}_{3}$, $\mathbf{B}^{-}=\mathbf{B}_{2}$ for $z\in\mathbb{U}_{2}$ and $\mathbf{B}^{-}=\mathbf{B}_{4}$ for $z\in\mathbb{U}_{4})$. Here, the matrix $\mathbf{L}(z)$ should be defined for each section of the contour based on the reflection coefficients.
\begin{equation}\label{4.3}
\begin{split}
\mathbf{L}_{1}(z)&=\begin{pmatrix}
    \Delta_{1}
     & \Delta_{2}
     & \Delta_{3}  \\
    -\alpha_{1} & 0 & \rho\chi_{3}  \\
    \alpha_{1}\alpha_{3}-\alpha_{2} & \alpha_{3} & -\rho\alpha_{3}\chi_{3}
  \end{pmatrix},   \quad
\mathbf{L}_{2}(z)=\begin{pmatrix}
    \chi_{2}\hat{\chi}_{2} & 0 & -\chi_{2}  \\
    0 & 0 & 0  \\
    \hat{\chi}_{2} & 0 & 0
  \end{pmatrix},  \\
\mathbf{L}_{3}(z)&=\begin{pmatrix}
    -\hat{\alpha}_{2}\hat{\chi}_{2} & \Delta_{4} & \hat{\alpha}_{2}  \\
    \Delta_{5} & \Delta_{6} & \Delta_{7}  \\
    \hat{\chi}_{2} & -\alpha_{3} & 0
  \end{pmatrix}, \quad
\mathbf{L}_{4}(z)=\begin{pmatrix}
    \alpha_{2}\hat{\alpha}_{2} & 0 & \hat{\alpha}_{2}  \\
    \Delta_{8} & 0 & \Delta_{9}  \\
    -\alpha_{2} & 0 & 0
  \end{pmatrix},
\end{split}
\end{equation}
with
\begin{equation}\label{4.4}
\begin{split}
\Delta_{1}&=(\alpha_{1}\alpha_{3}-\alpha_{2})\chi_{2}
+\frac{\mathrm{i}z}{q_{0}}\alpha_{1}\hat{\alpha}_{3}, \quad
\Delta_{2}=\alpha_{3}\chi_{2}+\displaystyle{\frac{\mathrm{i}z}{q_{0}}}\hat{\alpha}_{3}, \quad
\Delta_{4}=\alpha_{3}\hat{\alpha}_{2}-\displaystyle{\frac{\mathrm{i}z}{q_{0}}}\hat{\alpha}_{3},  \\
\Delta_{3}&=-(1+\rho\alpha_{3}\chi_{3})\chi_{2}
-\frac{\mathrm{i}z\rho}{q_{0}}\hat{\alpha}_{3}\chi_{3}, \quad
\Delta_{5}=\frac{\mathrm{i}q_{0}}{z}\hat{\alpha}_{1}\hat{\chi}_{2}
    +\frac{\mathrm{i}z\rho}{q_{0}}\hat{\chi}_{3}(1+\hat{\alpha}_{2}\hat{\chi}_{2}), \\
\Delta_{6}&=-\frac{\mathrm{i}q_{0}}{z}\alpha_{3}\hat{\alpha}_{1}
-\frac{z^{2}\rho}{q_{0}^{2}}\hat{\alpha}_{3}\hat{\chi}_{3}
-\frac{\mathrm{i}z\rho}{q_{0}}\alpha_{3}\hat{\alpha}_{2}\hat{\chi}_{3}, \quad
\Delta_{7}=-\frac{\mathrm{i}q_{0}}{z}\hat{\alpha}_{1}
-\frac{\mathrm{i}z\rho}{q_{0}}\hat{\alpha}_{2}\hat{\chi}_{3}, \\
\Delta_{8}&=\frac{\mathrm{i}z\rho}{q_{0}}\hat{\chi}_{3}(1-\alpha_{2}\hat{\alpha}_{2})
-\frac{\mathrm{i}q_{0}}{z}\alpha_{2}\hat{\alpha}_{1}-\alpha_{1}, \quad
\Delta_{9}=\rho\chi_{3}-\frac{\mathrm{i}q_{0}}{z}\hat{\alpha}_{1}
-\frac{\mathrm{i}z\rho}{q_{0}}\hat{\alpha}_{2}\hat{\chi}_{3},
\end{split}
\end{equation}
where $\alpha_{j}=\alpha_{j}(z)$, $\hat{\alpha}_{j}=\alpha_{j}(-q_{0}^{2}/z)$, $\rho=\rho(z)$, $\chi_{j}=\alpha_{j}^{*}(z^{*})$, $\hat{\chi}_{j}=\alpha_{j}^{*}(-q_{0}^{2}/z^{*})$ and $\mathbf{L}(z)=\mathbf{L}_{j}(z)$ for $z\in\mathbb{D}_{j}(j=1,\cdots,4)$.
\end{lemma}

To ensure the above RH problem has a unique solution, it is necessary to define an appropriate normalization condition as well. By considering the asymptotic behavior, it can be readily verified that
\begin{equation}\label{4.5}
\begin{split}
&\mathbf{B}(z)=\mathbf{B}_{\infty}+O(\frac{1}{z}), \quad (z\rightarrow \infty), \quad
\mathbf{B}(z)=\frac{\mathrm{i}}{z}\mathbf{B}_{0}+O(1), \quad (z\rightarrow 0),
\end{split}
\end{equation}
where
\begin{equation}\label{4.6}
\begin{split}
\mathbf{B}_{\infty}+\frac{\mathrm{i}}{z}\mathbf{B}_{0}=\mathbf{Y}_{+}(z), \quad
\mathbf{B}_{\infty}=\begin{pmatrix}
    1 & 0 & 0  \\
    \mathbf{0_{2\times1}} & \displaystyle{\frac{\mathbf{q}_{+}^{\perp}}{q_{0}}} & \displaystyle{\frac{\mathbf{q}_{+}}{q_{0}}}
  \end{pmatrix}, \quad
\mathbf{B}_{0}=\begin{pmatrix}
    0 & 0 & -q_{0}  \\
    -\mathbf{q}_{+} & \mathbf{0_{2\times1}} & \mathbf{0_{2\times1}}
  \end{pmatrix}.
\end{split}
\end{equation}
Besides the asymptotic behavior mentioned in~\eqref{4.1}, to completely define the RH problem described in~\eqref{4.2}, one needs to provide the residue conditions.

For simplicity, we denote by $\mathbf{B}_{-1,\omega}^{\pm}(x,t)$ the residue of $\mathbf{B}^{\pm}$ at $z=\omega$ and the mark $\mathbf{B}^{\pm}(x,t,z)=(\mathbf{b}_{1}^{\pm},\mathbf{b}_{2}^{\pm},\mathbf{b}_{3}^{\pm})$. By using the meromorphic matrices $\mathbf{B}^{\pm}$ in Lemma~\ref{lem:1}, the residue conditions can be obtained
\begin{subequations}\label{4.7}
\begin{align}
\mathbf{B}_{-1,z_{m}}^{+}(x,t)&=D_{m}[\mathbf{b}_{3}^{+}(z_{m}),\mathbf{0},\mathbf{0}], \quad
\mathbf{B}_{-1,-q_{0}^{2}/z_{m}^{*}}^{+}(x,t)=\frac{\mathrm{i}z_{m}^{*}}{q_{0}}\check{D}_{m}
[\mathbf{b}_{1}^{-}(z_{m}^{*}),\mathbf{0},\mathbf{0}], \\
\mathbf{B}_{-1,z_{m}^{*}}^{-}(x,t)&=\hat{D}_{m}[\mathbf{0},\mathbf{0},\mathbf{b}_{1}^{-}(z_{m}^{*})], \quad
\mathbf{B}_{-1,-q_{0}^{2}/z_{m}}^{-}(x,t)=\frac{\mathrm{i}z_{m}}{q_{0}}\bar{D}_{m}
[\mathbf{0},\mathbf{0},\mathbf{b}_{3}^{+}(z_{m})], \\
\mathbf{B}_{-1,\theta_{m}}^{+}(x,t)&=G_{m}[\mathbf{0},\mathbf{0},\mathbf{b}_{2}^{+}(\theta_{m})], \quad
\mathbf{B}_{-1,-q_{0}^{2}/\theta_{m}^{*}}^{+}(x,t)=\frac{\mathrm{i}\theta_{m}^{*}}{q_{0}}\check{G}_{m}
[\mathbf{0},\mathbf{b}_{3}^{-}(\theta_{m}^{*}),\mathbf{0}], \\
\mathbf{B}_{-1,\theta_{m}^{*}}^{-}(x,t)&=\hat{G}_{m}
[\mathbf{0},\mathbf{b}_{3}^{-}(\theta_{m}^{*}),\mathbf{0}], \quad
\mathbf{B}_{-1,-q_{0}^{2}/\theta_{m}}^{-}(x,t)=\bar{G}_{m}
[\mathbf{b}_{2}^{+}(\theta_{m}),\mathbf{0},\mathbf{0}],  \\
\mathbf{B}_{-1,\xi_{m}}^{+}(x,t)&=F_{m}[\mathbf{b}_{2}^{+}(\xi_{m}),\mathbf{0},\mathbf{0}], \quad
\mathbf{B}_{-1,-q_{0}^{2}/\xi_{m}^{*}}^{+}(x,t)=\frac{\mathrm{i}\xi_{m}^{*}}{q_{0}}\check{F}_{m}
[\mathbf{0},\mathbf{b}_{1}^{-}(\xi_{m}^{*}),\mathbf{0}], \\
\mathbf{B}_{-1,\xi_{m}^{*}}^{-}(x,t)&=\hat{F}_{m}
[\mathbf{0},\mathbf{b}_{1}^{-}(\xi_{m}^{*}),\mathbf{0}], \quad
\mathbf{B}_{-1,-q_{0}^{2}/\xi_{m}}^{-}(x,t)=\bar{F}_{m}
[\mathbf{0},\mathbf{0},\mathbf{b}_{2}^{+}(\xi_{m})],
\end{align}
\end{subequations}
with norming constants
\begin{subequations}\label{4.8}
\begin{align}
D_{m}&=\frac{d_{m}}{s_{11}^{\prime}(z_{m})}\mathrm{e}^{\mathrm{i}[\delta_{3}(z_{m})
-\delta_{1}(z_{m})]}, \quad
\check{D}_{m}=\frac{\check{d}_{m}}{h_{33}^{\prime}(-q_{0}^{2}/z_{m}^{*})}
\mathrm{e}^{\mathrm{i}[\delta_{1}(z_{m}^{*})-\delta_{3}(z_{m}^{*})]}, \\
\hat{D}_{m}&=\frac{\hat{d}_{m}}{h_{11}^{\prime}(z_{m}^{*})}
\mathrm{e}^{\mathrm{i}[\delta_{1}(z_{m}^{*})-\delta_{3}(z_{m}^{*})]}, \quad
\bar{D}_{m}=\frac{\bar{d}_{m}h_{22}(z_{m})}{s_{33}^{\prime}(-q_{0}^{2}/z_{m})}
\mathrm{e}^{\mathrm{i}[\delta_{3}(z_{m})-\delta_{1}(z_{m})]}, \\
G_{m}&=\frac{g_{m}}{h_{22}^{\prime}(\theta_{m})}
\mathrm{e}^{\mathrm{i}[\delta_{2}(\theta_{m})-\delta_{3}(\theta_{m})]}, \quad
\check{G}_{m}=\frac{\check{g}_{m}h_{11}(\theta_{m}^{*})}{s_{22}^{\prime}(-q_{0}^{2}/\theta_{m}^{*})}
\mathrm{e}^{\mathrm{i}[\delta_{3}(\theta_{m}^{*})-\delta_{2}(\theta_{m}^{*})]}, \\
\hat{G}_{m}&=\frac{\hat{g}_{m}h_{11}(\theta_{m}^{*})}{s_{22}^{\prime}(\theta_{m}^{*})}
\mathrm{e}^{\mathrm{i}[\delta_{3}(\theta_{m}^{*})-\delta_{2}(\theta_{m}^{*})]}, \quad
\bar{G}_{m}=\frac{\bar{g}_{m}}{h_{22}^{\prime}(-q_{0}^{2}/\theta_{m})}
\mathrm{e}^{\mathrm{i}[\delta_{2}(\theta_{m})-\delta_{3}(\theta_{m})]}, \\
F_{n}&=\frac{f_{m}}{s_{11}^{\prime}(\xi_{m})}
\mathrm{e}^{\mathrm{i}[\delta_{2}(\xi_{m})-\delta_{1}(\xi_{m})]}, \quad
\check{F}_{m}=\frac{\check{f}_{m}}{s_{22}^{\prime}(-q_{0}^{2}/\xi_{m}^{*})}
\mathrm{e}^{\mathrm{i}[\delta_{1}(\xi_{m}^{*})-\delta_{2}(\xi_{m}^{*})]}, \\
\hat{F}_{m}&=\frac{\hat{f}_{m}}{s_{22}^{\prime}(\xi_{m}^{*})}
\mathrm{e}^{\mathrm{i}[\delta_{1}(\xi_{m}^{*})-\delta_{2}(\xi_{m}^{*})]}, \quad
\bar{F}_{m}=\frac{\bar{f}_{m}}{s_{33}^{\prime}(-q_{0}^{2}/\xi_{m})}
\mathrm{e}^{\mathrm{i}[\delta_{2}(\xi_{m})-\delta_{1}(\xi_{m})]},
\end{align}
\end{subequations}
then the defined norming constants satisfy the following symmetric relations
\begin{subequations}\label{4.10}
\begin{align}
\hat{D}_{m}&=-D_{m}^{*}, \quad
\check{D}_{m}=-\frac{q_{0}^{2}}{(z_{m}^{*})^{2}}D_{m}^{*}, \quad
\bar{D}_{m}=\frac{q_{0}^{2}}{z_{m}^{2}}D_{m},  \\
\hat{G}_{m}&=-\frac{G_{m}^{*}}{\rho(\theta_{m}^{*})}, \quad
\check{G}_{m}=\frac{\mathrm{i}q_{0}^{3}}{(\theta_{m}^{*})^{3}}
\frac{G_{m}^{*}}{\rho(\theta_{m}^{*})}, \quad
\bar{G}_{m}=\frac{\mathrm{i}q_{0}}{\theta_{m}}G_{m}, \\
\hat{F}_{m}&=-\frac{F_{m}^{*}}{\rho(\xi_{m}^{*})}, \quad
\check{F}_{m}=\frac{\mathrm{i}q_{0}^{3}}{(\xi_{m}^{*})^{3}}
\frac{F_{m}^{*}}{\rho(\xi_{m}^{*})}, \quad
\bar{F}_{m}=\frac{\mathrm{i}q_{0}}{\xi_{m}}F_{m}.
\end{align}
\end{subequations}

\subsection{Formal solutions of the Riemann-Hilbert problem and reconstruction formula}
\label{s:Inverse problem:Formal solutions}

\begin{theorem}\label{thm:3}
Pure soliton solutions of the RHP defined by Lemma~\ref{lem:1} are given by the system of matrix algebraic-integral equations
\begin{equation}\label{4.11}
\begin{split}
\mathbf{B}(x,t,z)&=\mathbf{Y}_{+}(z)+\frac{\mathrm{i}}{2\pi}
\int_{\Sigma}\frac{\mathbf{B}^{-}(\zeta)\bar{\mathbf{L}}(\zeta)}{\zeta-z}\mathrm{d}\zeta
+\sum_{m=1}^{\mathrm{M}}\left[ \frac{\mathbf{B}_{-1,\varepsilon_{m}}^{+}}{z-\varepsilon_{m}} +\frac{\mathbf{B}_{-1,\varepsilon_{m}^{*}}^{-}}{z-\varepsilon_{m}^{*}}
+\frac{\mathbf{B}_{-1,-q_{0}^{2}/\varepsilon_{m}^{*}}^{+}}{z+(q_{0}^{2}/\varepsilon_{m}^{*})}
+\frac{\mathbf{B}_{-1,-q_{0}^{2}/\varepsilon_{m}}^{-}}{z+(q_{0}^{2}/\varepsilon_{m})} \right],
\end{split}
\end{equation}
where $\bar{\mathbf{L}}(z)=\mathbf{e}^{\mathrm{i}\mathbf{\Omega}}
\mathbf{L}(z)\mathbf{e}^{-\mathrm{i}\mathbf{\Omega}}$, with $\{\varepsilon_{m}\}_{m=1}^{\mathrm{M}}$ denotes the set of all discrete eigenvalues and $\mathbf{B}(x,t,z)=\mathbf{B}^{\pm}(x,t,z)$ for $z\in\mathbb{U}^{\pm}$. Moreover, the eigenfunctions in the residue conditions~\eqref{4.7} are given by
\begin{equation}\label{4.12}
\begin{split}
\mathbf{b}_{2}^{+}(z)&=\sum_{m=1}^{\mathrm{M}_{2}}\left[ \frac{\hat{G}_{m}}{z-\theta_{m}^{*}}
+\frac{\mathrm{i}\theta_{m}^{*}}{q_{0}}\frac{\check{G}_{m}}{z+(q_{0}^{2}/\theta_{m}^{*})} \right]\mathbf{b}_{3}^{-}(\theta_{m}^{*})+\frac{\mathrm{i}}{2\pi}
\int_{\Sigma}\frac{[\mathbf{B}^{-}(\zeta)\bar{\mathbf{L}}(\zeta)]_{2}}{\zeta-z}\mathrm{d}\zeta \\
&+\begin{pmatrix}
     0   \\
    \mathbf{q}_{+}^{\perp}/q_{0}
  \end{pmatrix}
+\sum_{m=1}^{\mathrm{M}_{3}}\left[ \frac{\hat{F}_{m}}{z-\xi_{m}^{*}}+\frac{\mathrm{i}\xi_{m}^{*}}{q_{0}}
\frac{\check{F}_{m}}{z+(q_{0}^{2}/\xi_{m}^{*})} \right]\mathbf{b}_{1}^{-}(\xi_{m}^{*}),
\end{split}
\end{equation}
where $z=\theta_{j^{\prime}}$ and $z=\xi_{l^{\prime}}$,
\begin{equation}\label{4.13}
\begin{split}
\mathbf{b}_{1}^{-}(z)&=\sum_{m=1}^{\mathrm{M}_{1}}\left[\frac{D_{m}\mathbf{b}_{3}^{+}(z_{m})}{z-z_{m}}
+\frac{\mathrm{i}z_{m}^{*}}{q_{0}}
\frac{\check{D}_{m}\mathbf{b}_{1}^{-}(z_{m}^{*})}{z+(q_{0}^{2}/z_{m}^{*})} \right]+\frac{\mathrm{i}}{2\pi}
\int_{\Sigma}\frac{[\mathbf{B}^{-}(\zeta)\bar{\mathbf{L}}(\zeta)]_{1}}{\zeta-z}\mathrm{d}\zeta  \\
&+\begin{pmatrix}
     1   \\
    -\mathrm{i}\mathbf{q}_{+}/z
  \end{pmatrix}
+\sum_{m=1}^{\mathrm{M}_{2}}\frac{\bar{G}_{m}\mathbf{b}_{2}^{+}(\theta_{m})}{z+(q_{0}^{2}/\theta_{m})}
+\sum_{m=1}^{\mathrm{M}_{3}}\frac{F_{m}\mathbf{b}_{2}^{+}(\xi_{m})}{z-\xi_{m}},
\end{split}
\end{equation}
where $z=z_{i^{\prime}}^{*}$ and $z=\xi_{l^{\prime}}^{*}$,
\begin{equation}\label{4.14}
\begin{split}
\mathbf{b}_{3}^{-}(\theta_{j^{\prime}}^{*})&=\sum_{m=1}^{\mathrm{M}_{1}}
\left[\frac{\hat{D}_{m}\mathbf{b}_{1}^{-}(z_{m}^{*})}{\theta_{j^{\prime}}^{*}-z_{m}^{*}}
+\frac{\mathrm{i}z_{n}}{q_{0}}\frac{\bar{D}_{m}
\mathbf{b}_{3}^{+}(z_{m})}{\theta_{j^{\prime}}^{*}+(q_{0}^{2}/z_{m})} \right]+\frac{\mathrm{i}}{2\pi}
\int_{\Sigma}\frac{[\mathbf{B}^{-}(\zeta)\bar{\mathbf{L}}(\zeta)]_{3}}{\zeta-\theta_{j^{\prime}}^{*}}
\mathrm{d}\zeta \\
&+\begin{pmatrix}
    -\mathrm{i}q_{0}/\theta_{j^{\prime}}^{*}   \\
    \mathbf{q}_{+}/q_{0}
  \end{pmatrix}
+\sum_{m=1}^{\mathrm{M}_{2}}\frac{G_{m}\mathbf{b}_{2}^{+}(\theta_{m})}{\theta_{j^{\prime}}^{*}-\theta_{m}}
+\sum_{m=1}^{\mathrm{M}_{3}}\frac{\bar{F}_{m}\mathbf{b}_{2}^{+}(\xi_{m})}
{\theta_{j^{\prime}}^{*}+(q_{0}^{2}/\xi_{m})},
\end{split}
\end{equation}
and
\begin{equation}\label{4.15}
\begin{split}
\mathbf{b}_{3}^{+}(z_{i^{\prime}})&=\sum_{m=1}^{\mathrm{M}_{1}}
\left[\frac{\hat{D}_{m}\mathbf{b}_{1}^{-}(z_{m}^{*})}{z_{i^{\prime}}-z_{m}^{*}}
+\frac{\mathrm{i}z_{m}}{q_{0}}\frac{\bar{D}_{m}
\mathbf{b}_{3}^{+}(z_{m})}{z_{i^{\prime}}+(q_{0}^{2}/z_{m})} \right]+\frac{\mathrm{i}}{2\pi}
\int_{\Sigma}\frac{[\mathbf{B}^{-}(\zeta)\bar{\mathbf{L}}(\zeta)]_{3}}{\zeta-z_{i^{\prime}}}
\mathrm{d}\zeta  \\
&+\begin{pmatrix}
    -\mathrm{i}q_{0}/z_{i^{\prime}}  \\
    \mathbf{q}_{+}/q_{0}
  \end{pmatrix}
+\sum_{m=1}^{\mathrm{M}_{2}}\frac{G_{m}\mathbf{b}_{2}^{+}(\theta_{m})}{z_{i^{\prime}}-\theta_{m}}
+\sum_{m=1}^{\mathrm{M}_{3}}\frac{\bar{F}_{m}\mathbf{b}_{2}^{+}(\xi_{m})}{z_{i^{\prime}}+(q_{0}^{2}/\xi_{m})},
\end{split}
\end{equation}
with $i^{\prime}=1,\cdots,\mathrm{M}_{1}$, $j^{\prime}=1,\cdots,\mathrm{M}_{2}$ and $l^{\prime}=1,\cdots,\mathrm{M}_{3}$.
\end{theorem}

By analyzing the solutions~\eqref{4.11} of the regularized RH problem, the potential can be reconstructed. Then the asymptotics in~\eqref{3.13a} imply
\begin{equation}\label{4.18}
\begin{split}
q_{k}(x,t)&=\mathrm{i}\lim_{z\rightarrow\infty}[z \nu_{+,(k+1)1}(x,t,z)], \quad k=1,2.
\end{split}
\end{equation}

\begin{theorem}\label{thm:4}
Pure soliton solutions $\mathbf{q}(x,t)$ of the coupled LPD equation~\eqref{1.1} with NBCS is reconstructed as
\begin{equation}\label{4.19}
\begin{split}
q_{k}(x,t)&=\sum_{m=1}^{\mathrm{M}_{1}}\left[ \mathrm{i}D_{m}\mathbf{b}_{(k+1)3}^{+}(z_{m})
-\frac{z_{m}^{*}}{q_{0}}\check{D}_{m}\mathbf{b}_{(k+1)1}^{-}(z_{m}^{*}) \right]
+\frac{1}{2\pi}\int_{\Sigma}[\mathbf{B}^{-}(\zeta)\bar{\mathbf{L}}(\zeta)]_{1}\mathrm{d}\zeta  \\
&+q_{+,k}+\sum_{m=1}^{\mathrm{M}_{2}}\mathrm{i}\bar{G}_{m}\mathbf{b}_{(k+1)2}^{+}(\theta_{m})
+\sum_{m=1}^{\mathrm{M}_{3}}\mathrm{i}F_{m}\mathbf{b}_{(k+1)2}^{+}(\xi_{m}), \quad k=1,2.
\end{split}
\end{equation}
\end{theorem}

\subsection{Trace formulae}
\label{s:Inverse problem:Trace formulae}

The method of constructing trace formula is similar to the problem of studying the trace formula of Manakov system~\cite{12}. Recall that $h_{22}(z)$ and $s_{22}(z)$ are analytic in $\mathbb{U}^{+}$ and in $\mathbb{U}^{-}$, we can define
\begin{subequations}\label{4.20}
\begin{align}
\gamma^{+}(z)&=h_{22}(z)\mathrm{e}^{\mathrm{i}\Delta\delta}
\prod_{m=1}^{\mathrm{M}_{2}}\frac{z-\theta_{m}^{*}}{z-\theta_{m}}
\frac{z+(q_{0}^{2}/\theta_{m}^{*})}{z+(q_{0}^{2}/\theta_{m})}
\prod_{m=1}^{\mathrm{M}_{3}}\frac{z-\xi_{m}^{*}}{z-\xi_{m}}
\frac{z+(q_{0}^{2}/\xi_{m}^{*})}{z+(q_{0}^{2}/\xi_{m})}, \quad z\in\mathbb{U}^{+}, \label{4.20a} \\
\gamma^{-}(z)&=\frac{1}{s_{22}(z)}\mathrm{e}^{\mathrm{i}\Delta\delta}
\prod_{m=1}^{\mathrm{M}_{2}}\frac{z-\theta_{m}^{*}}{z-\theta_{m}}
\frac{z+(q_{0}^{2}/\theta_{m}^{*})}{z+(q_{0}^{2}/\theta_{m})}
\prod_{m=1}^{\mathrm{M}_{3}}\frac{z-\xi_{m}^{*}}{z-\xi_{m}}
\frac{z+(q_{0}^{2}/\xi_{m}^{*})}{z+(q_{0}^{2}/\xi_{m})}, \quad  z\in\mathbb{U}^{-}. \label{4.20b}
\end{align}
\end{subequations}
with $\gamma^{\pm}(z)$ are analytic in $\mathbb{U}^{\pm}$ and $\Delta\delta=\delta_{+}-\delta_{-}$. For $z\in\Sigma$
\begin{equation}\label{4.21}
\begin{split}
\ln h_{22}(z)-\ln \frac{1}{s_{22}(z)}=-\ln \left[1+\rho(z)\alpha_{3}(z)\alpha_{3}^{*}(z^{*})
+\frac{z^{2}\rho(z)}{q_{0}^{2}}\alpha_{3}(-\frac{q_{0}^{2}}{z})\alpha_{3}^{*}(-\frac{q_{0}^{2}}{z^{*}}) \right],
\end{split}
\end{equation}

By employing suitable cofactor expansions and considering the definition~\eqref{2.79}, then we obtain
\begin{subequations}\label{4.22}
\begin{align}
\ln s_{11}(z)-\ln \frac{1}{h_{11}(z)}&=-\ln \left[1+\frac{\alpha_{1}(z)\alpha_{1}^{*}(z^{*})}{\rho(z)}
+\alpha_{2}(z)\alpha_{2}^{*}(z^{*}) \right], \\
\ln h_{33}(z)-\ln \frac{1}{s_{33}(z)}&=-\ln \left[1+\alpha_{2}(-\frac{q_{0}^{2}}{z})\alpha_{2}^{*}(-\frac{q_{0}^{2}}{z^{*}})
+\frac{q_{0}^{2}}{z^{2}\rho(z)}\alpha_{1}(-\frac{q_{0}^{2}}{z})\alpha_{1}^{*}(-\frac{q_{0}^{2}}{z^{*}}) \right],
\end{align}
\end{subequations}
where $z\in\Sigma$, $s_{11}(z)$, $h_{11}(z)$, $h_{33}(z)$ and $s_{33}(z)$ are analytic only in the fundamental domains $\mathbb{U}_{1}$, $\mathbb{U}_{2}$, $\mathbb{U}_{3}$ and $\mathbb{U}_{4}$, respectively. Recalling that $\mathbf{H}(z)\mathbf{S}(z)=\mathbf{I}$, then
\begin{subequations}\label{4.23}
\begin{align}
\ln h_{33}(z)-\ln \frac{1}{h_{11}(z)}&=\ln s_{22}(z)-\ln
\left[1-\alpha_{2}^{*}(z^{*})\alpha_{2}^{*}(-\frac{q_{0}^{2}}{z^{*}}) \right], \quad z\in\Sigma, \\
\ln s_{11}(z)-\ln \frac{1}{s_{33}(z)}&=\ln h_{22}(z)-\ln
\left[1-\alpha_{2}(z)\alpha_{2}(-\frac{q_{0}^{2}}{z}) \right], \quad z\in\Sigma.
\end{align}
\end{subequations}

Below we define
\begin{align}\label{4.24}
\begin{split}
\hat{\gamma}^{+}(z)= \left \{
\begin{array}{ll}
    \gamma_{1}(z),   & z\in\mathbb{U}_{1}, \\
    \gamma_{3}(z),   & z\in\mathbb{U}_{3},
\end{array}
\right. \quad
\hat{\gamma}^{-}(z)= \left \{
\begin{array}{ll}
    \gamma_{2}(z),   & z\in\mathbb{U}_{2}, \\
    \gamma_{4}(z),   & z\in\mathbb{U}_{4},
\end{array}
\right.
\end{split}
\end{align}
where
\begin{subequations}\label{4.25}
\begin{align}
\gamma_{1}(z)&=s_{11}(z)\prod_{m=1}^{\mathrm{M}_{1}}\frac{z-z_{m}^{*}}{z-z_{m}}
\frac{z+(q_{0}^{2}/z_{m})}{z+(q_{0}^{2}/z_{m}^{*})}
\prod_{m=1}^{\mathrm{M}_{2}}\frac{z+(q_{0}^{2}/\theta_{m}^{*})}{z+(q_{0}^{2}/\theta_{m})}
\prod_{m=1}^{\mathrm{M}_{3}}\frac{z-\xi_{m}^{*}}{z-\xi_{m}}, \quad  z\in\mathbb{U}_{1}, \label{4.25a} \\
\gamma_{2}(z)&=\frac{1}{h_{11}(z)}\prod_{m=1}^{\mathrm{M}_{1}}\frac{z-z_{m}^{*}}{z-z_{m}}
\frac{z+(q_{0}^{2}/z_{m})}{z+(q_{0}^{2}/z_{m}^{*})}
\prod_{m=1}^{\mathrm{M}_{2}}\frac{z+(q_{0}^{2}/\theta_{m}^{*})}{z+(q_{0}^{2}/\theta_{m})}
\prod_{m=1}^{\mathrm{M}_{3}}\frac{z-\xi_{m}^{*}}{z-\xi_{m}}, \quad  z\in\mathbb{U}_{2}, \label{4.25b} \\
\gamma_{3}(z)&=\frac{h_{33}(z)}{\mathrm{e}^{\mathrm{i}\Delta\delta}}
\prod_{m=1}^{\mathrm{M}_{1}}\frac{z-z_{m}^{*}}{z-z_{m}}\frac{z+(q_{0}^{2}/z_{m})}{z+(q_{0}^{2}/z_{m}^{*})}
\prod_{m=1}^{\mathrm{M}_{2}}\frac{z-\theta_{m}}{z-\theta_{m}^{*}}
\prod_{m=1}^{\mathrm{M}_{3}}\frac{z+(q_{0}^{2}/\xi_{m})}{z+(q_{0}^{2}/\xi_{m}^{*})}, \quad  z\in\mathbb{U}_{3},
\label{4.25c} \\
\gamma_{4}(z)&=\frac{\mathrm{e}^{-\mathrm{i}\Delta\delta}}{s_{33}(z)}
\prod_{m=1}^{\mathrm{M}_{1}}\frac{z-z_{m}^{*}}{z-z_{m}}\frac{z+(q_{0}^{2}/z_{m})}{z+(q_{0}^{2}/z_{m}^{*})}
\prod_{m=1}^{\mathrm{M}_{2}}\frac{z-\theta_{m}}{z-\theta_{m}^{*}}
\prod_{m=1}^{\mathrm{M}_{3}}\frac{z+(q_{0}^{2}/\xi_{m})}{z+(q_{0}^{2}/\xi_{m}^{*})}, \quad z\in\mathbb{U}_{4}.
\label{4.25d}
\end{align}
\end{subequations}
Now $\hat{\gamma}^{\pm}(z)$ are analytic in $\mathbb{D}^{\pm}$.

Equations~\eqref{4.22} and~\eqref{4.23} are simplified into the following form by means of~\eqref{4.25}
\begin{subequations}
\begin{align}
\ln \gamma_{1}(z)-\ln \gamma_{2}(z)&=J_{1}(z), \quad
\ln \gamma_{3}(z)-\ln \gamma_{2}(z)=J_{2}(z), \quad  z\in\Sigma, \label{4.26a} \\
\ln \gamma_{3}(z)-\ln \gamma_{4}(z)&=J_{3}(z), \quad
\ln \gamma_{1}(z)-\ln \gamma_{4}(z)=J_{4}(z), \quad  z\in\Sigma, \label{4.26b}
\end{align}
\end{subequations}
with
\begin{subequations}\label{4.27}
\begin{align}
J_{1}(z)&=-\ln \left[1+\frac{\alpha_{1}(z)\alpha_{1}^{*}(z^{*})}{\rho(z)}
+\alpha_{2}(z)\alpha_{2}^{*}(z^{*}) \right], \\
J_{2}(z)&=\frac{1}{2\pi\mathrm{i}}\int_{\mathbb{R}}\frac{J_{0}(\zeta)}{\zeta-z}\mathrm{d}\zeta-\ln
\left[1-\alpha_{2}^{*}(z^{*})\alpha_{2}^{*}(-\frac{q_{0}^{2}}{z^{*}}) \right], \\
J_{3}(z)&=-\ln \left[1+\alpha_{2}(-\frac{q_{0}^{2}}{z})\alpha_{2}^{*}(-\frac{q_{0}^{2}}{z^{*}})
+\frac{q_{0}^{2}}{z^{2}\rho(z)}\alpha_{1}(-\frac{q_{0}^{2}}{z})
\alpha_{1}^{*}(-\frac{q_{0}^{2}}{z^{*}}) \right], \\
J_{4}(z)&=-\frac{1}{2\pi\mathrm{i}}\int_{\mathbb{R}}\frac{J_{0}(\zeta)}{\zeta-z}
\mathrm{d}\zeta-\ln
\left[1-\alpha_{2}(z)\alpha_{2}(-\frac{q_{0}^{2}}{z}) \right],
\end{align}
\end{subequations}
and
\begin{equation}\label{4.28}
\begin{split}
J_{0}(z)=\ln \left[1+\rho(z)\alpha_{3}(z)\alpha_{3}^{*}(z^{*})
+\frac{z^{2}\rho(z)}{q_{0}^{2}}\alpha_{3}(-\frac{q_{0}^{2}}{z})\alpha_{3}^{*}(-\frac{q_{0}^{2}}{z^{*}}) \right],
\end{split}
\end{equation}
where $J_{j}(z)(j=1,2,3,4)$ are jump conditions. Then
\begin{equation}\label{4.29}
\begin{split}
\ln \hat{\gamma}^{+}(z)-\ln \hat{\gamma}^{-}(z)=J_{i}(z), \quad  z\in\Sigma_{i},
\end{split}
\end{equation}
for $i=1,2,3,4$ and where $\Sigma_{i}$ is the boundary of $\bar{\mathbb{U}}_{i}\cap\bar{\mathbb{U}}_{i+1 \bmod 4}$.

Using the Plemelj's formulae then yields
\begin{equation}\label{4.30}
\begin{split}
\ln \hat{\gamma}^{\pm}(z)=\frac{1}{2\pi\mathrm{i}}\int_{\Sigma}\frac{J(\zeta)}{\zeta-z}
\mathrm{d}\zeta, \quad z\in\mathbb{C} \backslash \Sigma.
\end{split}
\end{equation}
Substituting~\eqref{4.30} for~\eqref{4.25a}, we then obtain the analytic scattering coefficients
\begin{equation}\label{4.31}
\begin{split}
s_{11}(z)&=\exp\left[ \frac{1}{2\pi\mathrm{i}}\int_{\Sigma}\frac{J(\zeta)}{\zeta-z}\mathrm{d}\zeta \right] \prod_{m=1}^{\mathrm{M}_{1}}\frac{z-z_{m}}{z-z_{m}^{*}}\frac{z
+\displaystyle{\frac{q_{0}^{2}}{z_{m}^{*}}}}{z+\displaystyle{\frac{q_{0}^{2}}{z_{m}}}}
\prod_{m=1}^{\mathrm{M}_{2}}\frac{z+\displaystyle{\frac{q_{0}^{2}}{\theta_{m}}}}{z
+\displaystyle{\frac{q_{0}^{2}}{\theta_{m}^{*}}}}
\prod_{m=1}^{\mathrm{M}_{3}}\frac{z-\xi_{m}}{z-\xi_{m}^{*}}, \quad z\in\mathbb{U}_{1}.
\end{split}
\end{equation}
Thus
\begin{equation}\label{4.32}
\begin{split}
\ln \gamma^{\pm}(z)=\frac{\mathrm{i}}{2\pi}\int_{\mathbb{R}}\frac{J_{0}(\zeta)}{\zeta-z}
\mathrm{d}\zeta, \quad z\in\mathbb{C} \backslash \Sigma.
\end{split}
\end{equation}

Substituting~\eqref{4.32} for~\eqref{4.20a}, we then obtain the analytic scattering coefficients
\begin{equation}\label{4.33}
\begin{split}
h_{22}(z)&=\exp\left[ \frac{1}{2\pi\mathrm{i}}\int_{\mathbb{R}}\frac{J_{0}(\zeta)}{z
-\zeta}\mathrm{d}\zeta-\mathrm{i}\Delta\delta \right]\prod_{m=1}^{\mathrm{M}_{2}}\frac{z-\theta_{m}}{z-\theta_{m}^{*}}\frac{z
+\displaystyle{\frac{q_{0}^{2}}{\theta_{m}}}}{z+\displaystyle{\frac{q_{0}^{2}}{\theta_{m}^{*}}}}
\prod_{m=1}^{\mathrm{M}_{3}}\frac{z-\xi_{m}}{z-\xi_{m}^{*}}
\frac{z+\displaystyle{\frac{q_{0}^{2}}{\xi_{m}}}}{z+\displaystyle{\frac{q_{0}^{2}}{\xi_{m}^{*}}}},
\end{split}
\end{equation}
for $z\in\mathbb{U}^{+}$. Then the asymptotic phase difference $\Delta\delta=\delta_{+}-\delta_{-}$ is obtained
\begin{equation}\label{4.34}
\begin{split}
\Delta \delta=\frac{1}{2\pi}\int_{\Sigma}\frac{J(\zeta)}{\zeta}\mathrm{d}\zeta
-4\sum_{m=1}^{\mathrm{M}_{1}}\arg z_{m}+2\sum_{m=1}^{\mathrm{M}_{2}}\arg \theta_{m}-2\sum_{m=1}^{\mathrm{M}_{3}}\arg \xi_{m}.
\end{split}
\end{equation}

\section{Pure soliton solutions}
\label{s:Pure soliton solutions}

\begin{theorem}\label{thm:5}
Pure soliton solutions~\eqref{4.19} of the coupled LPD equation~\eqref{1.1} with NBCS may be written
\begin{equation}\label{5.1}
\begin{split}
\mathbf{q}(x,t)=\frac{1}{\operatorname{det} \mathbf{K}}\begin{pmatrix}
\operatorname{det} \mathbf{K}_{1}^{\text {aug }} \\
\operatorname{det} \mathbf{K}_{2}^{\text {aug }}
\end{pmatrix},  \quad
\mathbf{K}_{k}^{\mathrm{aug}}=\left(\begin{array}{ll}
q_{+, k} & \mathbf{\eta}^{T} \\
\mathbf{w}_{k} & \mathbf{K}
\end{array}\right), \quad k=1,2,
\end{split}
\end{equation}
the vectors $\mathbf{w}_{k}$, $\eta$ and matrix $\mathbf{K}$ are
\begin{equation}\label{5.2}
\begin{split}
\mathbf{w}_{k}&=\left(w_{k1}, \cdots, w_{k(2\mathrm{M}_{1}+\mathrm{M}_{2}+\mathrm{M}_{3})}\right)^{T}, \quad \eta=\left(\eta_{1}, \cdots, \eta_{2\mathrm{M}_{1}+\mathrm{M}_{2}+\mathrm{M}_{3}}\right)^{T}, \quad \mathbf{K}=\mathbf{I}-\mathbf{P},
\end{split}
\end{equation}
where
\begin{align}\label{5.3}
\begin{split}
\eta_{j}(x, t)=\left\{\begin{array}{ll}
-\mathrm{i} D_{j},  &j=1, \cdots, \mathrm{M}_{1}, \\
\displaystyle{\frac{z_{j-\mathrm{M}_{1}}^{*}}{q_{0}}}\check{D}_{j-\mathrm{M}_{1}}, & j=\mathrm{M}_{1}+1, \cdots, 2 \mathrm{M}_{1}, \\
-\mathrm{i} \bar{G}_{j-2 \mathrm{M}_{1}}, & j=2\mathrm{M}_{1}+1, \cdots, 2 \mathrm{M}_{1}+\mathrm{M}_{2}, \\
-\mathrm{i} F_{j-2 \mathrm{M}_{1}-\mathrm{M}_{2}}, & j=2\mathrm{M}_{1}+\mathrm{M}_{2}+1, \cdots, 2 \mathrm{M}_{1}+\mathrm{M}_{2}+\mathrm{M}_{3},
\end{array}\right.
\end{split}
\end{align}
and
\begin{align}\label{5.4}
\begin{array}{l}
m_{kj}(x,t)=\left\{\begin{array}{ll}
\displaystyle{\frac{q_{+,k}}{q_{0}}}, \quad  j=1, \cdots, \mathrm{M}_{1}, \\
-\mathrm{i}\displaystyle{\frac{q_{+,k}}{z_{j-\mathrm{M}_{1}}^{*}}}, \quad j=\mathrm{M}_{1}+1, \cdots, 2\mathrm{M}_{1}, \\
\displaystyle{(-1)^{k+1}\frac{q_{+,\bar{k}}^{*}}{q_{0}}
+\frac{q_{+,k}}{q_{0}}\sum_{m=1}^{\mathrm{M}_{2}}C_{m}^{(1)}(\theta_{j-2\mathrm{M}_{1}})
-\mathrm{i}q_{+,k}\sum_{m=1}^{\mathrm{M}_{3}}\frac{C_{m}^{(2)}(\theta_{j-2\mathrm{M}_{1}})}{\xi_{m}^{*}}},  \\ j=2\mathrm{M}_{1}+1, \cdots, 2\mathrm{M}_{1}+\mathrm{M}_{2}, \\
\displaystyle{\frac{q_{+,k}}{q_{0}}\sum_{m=1}^{\mathrm{M}_{2}}C_{m}^{(1)}(\xi_{j-2\mathrm{M}_{1}-\mathrm{M}_{2}})
-\mathrm{i}q_{+,k}\sum_{m=1}^{\mathrm{M}_{3}}\frac{C_{m}^{(2)}(\xi_{j-2\mathrm{M}_{1}-\mathrm{M}_{2}})}{\xi_{m}^{*}}} & \\
+(-1)^{k+1}\displaystyle{\frac{q_{+,\bar{k}}^{*}}{q_{0}}}, \quad
j=2\mathrm{M}_{1}+\mathrm{M}_{2}+1, \cdots, 2\mathrm{M}_{1}+\mathrm{M}_{2}+\mathrm{M}_{3},
\end{array}\right.
\end{array}
\end{align}
with $k=1,2$, $\bar{k}=3-k$ and the coefficients $\mathbf{P}=P_{ij}(x,t)$ are defined as follows. When $i,j=1, \cdots, \mathrm{M}_{1}$,
\begin{equation}\label{5.5}
\begin{split}
P_{ij}&=C_{j}^{(4)}(z_{i}).
\end{split}
\end{equation}
When $i=1, \cdots, \mathrm{M}_{1}$ and $j=\mathrm{M}_{1}+1, \cdots, 2\mathrm{M}_{1}$,
\begin{equation}\label{5.6}
\begin{split}
P_{ij}&=C_{j-\mathrm{M}_{1}}^{(3)}(z_{i}).
\end{split}
\end{equation}
When $i=1, \cdots, \mathrm{M}_{1}$ and $j=2\mathrm{M}_{1}+1, \cdots, 2\mathrm{M}_{1}+\mathrm{M}_{2}$,
\begin{equation}\label{5.7}
\begin{split}
P_{ij}&=C_{j-2\mathrm{M}_{1}}^{(5)}(z_{i}).
\end{split}
\end{equation}
When $i=1, \cdots, \mathrm{M}_{1}$ and $j=2\mathrm{M}_{1}+\mathrm{M}_{2}+1, \cdots, 2 \mathrm{M}_{1}+\mathrm{M}_{2}+\mathrm{M}_{3}$,
\begin{equation}\label{5.8}
\begin{split}
P_{ij}&=C_{j-2\mathrm{M}_{1}-\mathrm{M}_{2}}^{(6)}(z_{i}).
\end{split}
\end{equation}
When $i=\mathrm{M}_{1}+1, \cdots, 2\mathrm{M}_{1}$ and $j=1, \cdots, \mathrm{M}_{1}$,
\begin{equation}\label{5.9}
\begin{split}
P_{ij}&=C_{j}^{(7)}(z_{i-\mathrm{M}_{1}}^{*}).
\end{split}
\end{equation}
When $i=\mathrm{M}_{1}+1, \cdots, 2\mathrm{M}_{1}$ and $j=\mathrm{M}_{1}+1, \cdots, 2\mathrm{M}_{1}$,
\begin{equation}\label{5.10}
\begin{split}
P_{ij}&=C_{j-\mathrm{M}_{1}}^{(8)}(z_{i-\mathrm{M}_{1}}^{*}).
\end{split}
\end{equation}
When $i=\mathrm{M}_{1}+1, \cdots, 2\mathrm{M}_{1}$ and $j=2\mathrm{M}_{1}+1, \cdots, 2\mathrm{M}_{1}+\mathrm{M}_{2}$,
\begin{equation}\label{5.11}
\begin{split}
P_{ij}&=C_{j-2\mathrm{M}_{1}}^{(9)}(z_{i-\mathrm{M}_{1}}^{*}).
\end{split}
\end{equation}
When $i=\mathrm{M}_{1}+1, \cdots, 2\mathrm{M}_{1}$ and $j=2\mathrm{M}_{1}+\mathrm{M}_{2}+1, \cdots, 2 \mathrm{M}_{1}+\mathrm{M}_{2}+\mathrm{M}_{3}$,
\begin{equation}\label{5.12}
\begin{split}
P_{ij}=C_{j-2\mathrm{M}_{1}-\mathrm{M}_{2}}^{(10)}(z_{i-\mathrm{M}_{1}}^{*}).
\end{split}
\end{equation}
When $i=2\mathrm{M}_{1}+1, \cdots, 2\mathrm{M}_{1}+\mathrm{M}_{2}$ and $j=1, \cdots, \mathrm{M}_{1}$,
\begin{equation}\label{5.13}
\begin{split}
P_{ij}&=\sum_{m=1}^{\mathrm{M}_{2}}C_{m}^{(1)}(\theta_{i-2\mathrm{M}_{1}})C_{j}^{(4)}(\theta_{m}^{*})
+\sum_{m=1}^{\mathrm{M}_{3}}C_{m}^{(2)}(\theta_{i-2\mathrm{M}_{1}})C_{j}^{(7)}(\xi_{m}^{*}).
\end{split}
\end{equation}
When $i=2\mathrm{M}_{1}+1, \cdots, 2\mathrm{M}_{1}+\mathrm{M}_{2}$ and $j=\mathrm{M}_{1}+1, \cdots, 2\mathrm{M}_{1}$,
\begin{equation}\label{5.14}
\begin{split}
P_{ij}&=\sum_{m=1}^{\mathrm{M}_{2}}C_{m}^{(1)}(\theta_{i-2\mathrm{M}_{1}})C_{j
-\mathrm{M}_{1}}^{(3)}(\theta_{m}^{*})+\sum_{m=1}^{\mathrm{M}_{3}}C_{m}^{(2)}(\theta_{i
-2\mathrm{M}_{1}})C_{j-\mathrm{M}_{1}}^{(8)}(\xi_{m}^{*}).
\end{split}
\end{equation}
When $i=2\mathrm{M}_{1}+1, \cdots, 2\mathrm{M}_{1}+\mathrm{M}_{2}$ and $j=2\mathrm{M}_{1}+1, \cdots, 2\mathrm{M}_{1}+\mathrm{M}_{2}$,
\begin{equation}\label{5.15}
\begin{split}
P_{ij}&=\sum_{m=1}^{\mathrm{M}_{2}}C_{m}^{(1)}(\theta_{i-2\mathrm{M}_{1}})C_{j
-2\mathrm{M}_{1}}^{(5)}(\theta_{m}^{*})+\sum_{m=1}^{\mathrm{M}_{3}}C_{m}^{(2)}(\theta_{i
-2\mathrm{M}_{1}})C_{j-2\mathrm{M}_{1}}^{(9)}(\xi_{m}^{*}).
\end{split}
\end{equation}
When $i=2\mathrm{M}_{1}+1, \cdots, 2\mathrm{M}_{1}+\mathrm{M}_{2}$ and $j=2\mathrm{M}_{1}+\mathrm{M}_{2}+1, \cdots, 2 \mathrm{M}_{1}+\mathrm{M}_{2}+\mathrm{M}_{3}$,
\begin{equation}\label{5.16}
\begin{split}
P_{ij}&=\sum_{m=1}^{\mathrm{M}_{2}}C_{m}^{(1)}(\theta_{i-2\mathrm{M}_{1}})C_{j-2\mathrm{M}_{1}
-\mathrm{M}_{2}}^{(6)}(\theta_{m}^{*})+\sum_{m=1}^{\mathrm{M}_{3}}C_{m}^{(2)}(\theta_{i
-2\mathrm{M}_{1}})C_{j-2\mathrm{M}_{1}-\mathrm{M}_{2}}^{(10)}(\xi_{m}^{*}).
\end{split}
\end{equation}
When $i=2\mathrm{M}_{1}+\mathrm{M}_{2}+1, \cdots, 2 \mathrm{M}_{1}+\mathrm{M}_{2}+\mathrm{M}_{3}$ and $j=1, \cdots, \mathrm{M}_{1}$,
\begin{equation}\label{5.17}
\begin{split}
P_{ij}&=\sum_{m=1}^{\mathrm{M}_{2}}C_{m}^{(1)}(\xi_{i-2\mathrm{M}_{1}
-\mathrm{M}_{2}})C_{j}^{(4)}(\theta_{m}^{*})+\sum_{m=1}^{\mathrm{M}_{3}}C_{m}^{(2)}(\xi_{i
-2\mathrm{M}_{1}-\mathrm{M}_{2}})C_{j}^{(7)}(\xi_{m}^{*}).
\end{split}
\end{equation}
When $i=2\mathrm{M}_{1}+\mathrm{M}_{2}+1, \cdots, 2 \mathrm{M}_{1}+\mathrm{M}_{2}+\mathrm{M}_{3}$ and $j=\mathrm{M}_{1}+1, \cdots, 2\mathrm{M}_{1}$,
\begin{equation}\label{5.18}
\begin{split}
P_{ij}&=\sum_{m=1}^{\mathrm{M}_{2}}C_{m}^{(1)}(\xi_{i-2\mathrm{M}_{1}-\mathrm{M}_{2}})C_{j
-\mathrm{M}_{1}}^{(3)}(\theta_{m}^{*})+\sum_{m=1}^{\mathrm{M}_{3}}C_{m}^{(2)}(\xi_{i
-2\mathrm{M}_{1}-\mathrm{M}_{2}})C_{j-\mathrm{M}_{1}}^{(8)}(\xi_{m}^{*}).
\end{split}
\end{equation}
When $i=2\mathrm{M}_{1}+\mathrm{M}_{2}+1, \cdots, 2 \mathrm{M}_{1}+\mathrm{M}_{2}+\mathrm{M}_{3}$ and $j=2\mathrm{M}_{1}+1, \cdots, 2\mathrm{M}_{1}+\mathrm{M}_{2}$,
\begin{equation}\label{5.19}
\begin{split}
P_{ij}&=\sum_{m=1}^{\mathrm{M}_{2}}C_{m}^{(1)}(\xi_{i-2\mathrm{M}_{1}-\mathrm{M}_{2}})C_{j
-2\mathrm{M}_{1}}^{(5)}(\theta_{m}^{*})+\sum_{m=1}^{\mathrm{M}_{3}}C_{m}^{(2)}(\xi_{i
-2\mathrm{M}_{1}-\mathrm{M}_{2}})C_{j-2\mathrm{M}_{1}}^{(9)}(\xi_{m}^{*}).
\end{split}
\end{equation}
When $i,j=2\mathrm{M}_{1}+\mathrm{M}_{2}+1, \cdots, 2 \mathrm{M}_{1}+\mathrm{M}_{2}+\mathrm{M}_{3}$,
\begin{equation}\label{5.20}
\begin{split}
P_{ij}=\sum_{m=1}^{\mathrm{M}_{2}}C_{m}^{(1)}(\xi_{i-2\mathrm{M}_{1}-\mathrm{M}_{2}})C_{j
-2\mathrm{M}_{1}-\mathrm{M}_{2}}^{(6)}(\theta_{m}^{*})+\sum_{m=1}^{\mathrm{M}_{3}}C_{m}^{(2)}(\xi_{i
-2\mathrm{M}_{1}-\mathrm{M}_{2}})C_{j-2\mathrm{M}_{1}-\mathrm{M}_{2}}^{(10)}(\xi_{m}^{*}).
\end{split}
\end{equation}
For simplicity, we define
\begin{equation}\label{5.21}
\begin{split}
C_{m}^{(1)}(z)&=\frac{\hat{G}_{m}}{z-\theta_{m}^{*}}
+\frac{\mathrm{i}\theta_{m}^{*}}{q_{0}}\frac{\check{G}_{m}}{z+(q_{0}^{2}/\theta_{m}^{*})}, \quad
C_{m}^{(4)}(z)=\frac{\mathrm{i}z_{n}}{q_{0}}\frac{\bar{D}_{m}}{z+(q_{0}^{2}/z_{n})}, \quad
C_{m}^{(7)}(z)=\frac{D_{m}}{z-z_{m}}, \\
C_{m}^{(2)}(z)&=\frac{\hat{F}_{m}}{z-\xi_{m}^{*}}
+\frac{\mathrm{i}\xi_{m}^{*}}{q_{0}}\frac{\check{F}_{m}}{z+(q_{0}^{2}/\xi_{m}^{*})}, \quad
C_{m}^{(8)}(z)=\frac{\mathrm{i}z_{m}^{*}}{q_{0}}\frac{\check{D}_{m}}{z+(q_{0}^{2}/z_{m}^{*})}, \quad
C_{m}^{(9)}(z)=\frac{\bar{G}_{m}}{z+(q_{0}^{2}/\theta_{m})}, \\
C_{m}^{(3)}(z)&=\frac{\hat{D}_{m}}{z-z_{m}^{*}}, \quad C_{m}^{(5)}(z)=\frac{G_{m}}{z-\theta_{m}}, \quad C_{m}^{(6)}(z)=\frac{\bar{F}_{m}}{z+(q_{0}^{2}/\xi_{m})}, \quad
C_{m}^{(10)}(z)=\frac{F_{m}}{z-\xi_{m}}.
\end{split}
\end{equation}
\end{theorem}

\begin{proof}
Then in the reflectionless case, by plugging in the appropriate eigenvalues for $k=1,2$, $\bar{k}=3-k$, the following algebraic equations are derived
\begin{equation}\label{5.22}
\begin{split}
\mathbf{b}_{(k+1)2}^{+}(\theta_{j^{\prime}})=(-1)^{k+1}\frac{q_{+,\bar{k}}^{*}}{q_{0}}
+\sum_{m=1}^{\mathrm{M}_{2}}C_{m}^{(1)}(\theta_{j^{\prime}})\mathbf{b}_{(k+1)3}^{-}(\theta_{m}^{*})
+\sum_{m=1}^{\mathrm{M}_{3}}C_{m}^{(2)}(\theta_{j^{\prime}})\mathbf{b}_{(k+1)1}^{-}(\xi_{m}^{*}),
\end{split}
\end{equation}
\begin{equation}\label{5.23}
\begin{split}
\mathbf{b}_{(k+1)2}^{+}(\xi_{l^{\prime}})=(-1)^{k+1}\frac{q_{+,\bar{k}}^{*}}{q_{0}}
+\sum_{m=1}^{\mathrm{M}_{2}}C_{m}^{(1)}(\xi_{l^{\prime}})\mathbf{b}_{(k+1)3}^{-}(\theta_{m}^{*})
+\sum_{m=1}^{\mathrm{M}_{3}}C_{m}^{(2)}(\xi_{l^{\prime}})\mathbf{b}_{(k+1)1}^{-}(\xi_{m}^{*}),
\end{split}
\end{equation}
\begin{equation}\label{5.24}
\begin{split}
\mathbf{b}_{(k+1)1}^{-}(z_{i^{\prime}}^{*})&=-\frac{\mathrm{i}q_{+,k}}{z_{i^{\prime}}^{*}}
+\sum_{m=1}^{\mathrm{M}_{1}}\left[ C_{m}^{(7)}(z_{i^{\prime}}^{*})\mathbf{b}_{(k+1)3}^{+}(z_{m})
+C_{m}^{(8)}(z_{i^{\prime}}^{*})\mathbf{b}_{(k+1)1}^{-}(z_{m}^{*}) \right]   \\
&+\sum_{m=1}^{\mathrm{M}_{2}}C_{m}^{(9)}(z_{i^{\prime}}^{*})\mathbf{b}_{(k+1)2}^{+}(\theta_{m})
+\sum_{m=1}^{\mathrm{M}_{3}}C_{m}^{(10)}(z_{i^{\prime}}^{*})\mathbf{b}_{(k+1)2}^{+}(\xi_{m}),
\end{split}
\end{equation}
\begin{equation}\label{5.25}
\begin{split}
\mathbf{b}_{(k+1)1}^{-}(\xi_{l^{\prime}}^{*})&=-\frac{\mathrm{i}q_{+,k}}{\xi_{l^{\prime}}^{*}}
+\sum_{m=1}^{\mathrm{M}_{1}}\left[ C_{m}^{(7)}(\xi_{l^{\prime}}^{*})\mathbf{b}_{(k+1)3}^{+}(z_{m})
+C_{m}^{(8)}(\xi_{l^{\prime}}^{*})\mathbf{b}_{(k+1)1}^{-}(z_{m}^{*}) \right]  \\
&+\sum_{m=1}^{\mathrm{M}_{2}}C_{m}^{(9)}(\xi_{l^{\prime}}^{*})\mathbf{b}_{(k+1)2}^{+}(\theta_{m})
+\sum_{m=1}^{\mathrm{M}_{3}}C_{m}^{(10)}(\xi_{l^{\prime}}^{*})\mathbf{b}_{(k+1)2}^{+}(\xi_{m}),
\end{split}
\end{equation}
\begin{equation}\label{5.26}
\begin{split}
\mathbf{b}_{(k+1)3}^{-}(\theta_{j^{\prime}}^{*})&=\frac{q_{+,k}}{q_{0}}+\sum_{m=1}^{\mathrm{M}_{1}}
\left[ C_{m}^{(3)}(\theta_{j^{\prime}}^{*})\mathbf{b}_{(k+1)1}^{-}(z_{m}^{*})
+C_{m}^{(4)}(\theta_{j^{\prime}}^{*})\mathbf{b}_{(k+1)3}^{+}(z_{m}) \right]  \\
&+\sum_{m=1}^{\mathrm{M}_{2}}C_{m}^{(5)}(\theta_{j^{\prime}}^{*})\mathbf{b}_{(k+1)2}^{+}(\theta_{m})
+\sum_{m=1}^{\mathrm{M}_{3}}C_{m}^{(6)}(\theta_{j^{\prime}}^{*})\mathbf{b}_{(k+1)2}^{+}(\xi_{m}),
\end{split}
\end{equation}
\begin{equation}\label{5.27}
\begin{split}
\mathbf{b}_{(k+1)3}^{-}(z_{i^{\prime}})&=\frac{q_{+,k}}{q_{0}}+\sum_{m=1}^{\mathrm{M}_{1}}
\left[ C_{m}^{(3)}(z_{i^{\prime}})\mathbf{b}_{(k+1)1}^{-}(z_{m}^{*})
+C_{m}^{(4)}(z_{i^{\prime}})\mathbf{b}_{(k+1)3}^{+}(z_{n}) \right]  \\
&+\sum_{m=1}^{\mathrm{M}_{2}}C_{m}^{(5)}(z_{i^{\prime}})\mathbf{b}_{(k+1)2}^{+}(\theta_{m})
+\sum_{m=1}^{\mathrm{M}_{3}}C_{m}^{(6)}(z_{i^{\prime}})\mathbf{b}_{(k+1)2}^{+}(\xi_{m}),
\end{split}
\end{equation}
where $i^{\prime}=1,\cdots,\mathrm{M}_{1}$, $j^{\prime}=1,\cdots,\mathrm{M}_{2}$ and $l^{\prime}=1,\cdots,\mathrm{M}_{3}$. Next, we obtain the following for $z=\theta_{j^{\prime}}$ and $z=\xi_{l^{\prime}}$
\begin{equation}\label{5.28}
\begin{split}
\mathbf{b}_{(k+1)2}^{+}(z)&=(-1)^{k+1}\frac{q_{+,\bar{k}}^{*}}{q_{0}}
+\frac{q_{+,k}}{q_{0}}\sum_{m=1}^{\mathrm{M}_{2}}C_{m}^{(1)}(z)-\mathrm{i}q_{+,k}\sum_{m=1}^{\mathrm{M}_{3}}
\frac{C_{m}^{(2)}(z)}{\xi_{m}^{*}}+\mathbf{\Lambda}_{3}+\mathbf{\Lambda}_{4} \\
&+\sum_{m=1}^{\mathrm{M}_{2}}\sum_{m^{\prime}=1}^{\mathrm{M}_{1}}C_{m}^{(1)}(z)
\left[ C_{m^{\prime}}^{(3)}(\theta_{m}^{*})\mathbf{b}_{(k+1)1}^{-}(z_{m^{\prime}}^{*})
+C_{m^{\prime}}^{(4)}(\theta_{m}^{*})\mathbf{b}_{(k+1)3}^{+}(z_{n^{\prime}})\right]  \\
&+\sum_{m=1}^{\mathrm{M}_{3}}\sum_{m^{\prime}=1}^{\mathrm{M}_{1}}C_{m}^{(2)}(z)
\left[ C_{m^{\prime}}^{(7)}(\xi_{m}^{*})\mathbf{b}_{(k+1)3}^{+}(z_{m^{\prime}})
+C_{m^{\prime}}^{(8)}(\xi_{m}^{*})\mathbf{b}_{(k+1)1}^{-}(z_{m^{\prime}}^{*})\right],
\end{split}
\end{equation}
where
\begin{equation}\label{5.29}
\begin{split}
\mathbf{\Lambda}_{3}=\sum_{m=1}^{\mathrm{M}_{2}}\sum_{m^{\prime}=1}^{\mathrm{M}_{2}}C_{m}^{(1)}(z)
C_{m^{\prime}}^{(5)}(\theta_{m}^{*})\mathbf{b}_{(k+1)2}^{+}(\theta_{m^{\prime}})
+\sum_{m=1}^{\mathrm{M}_{2}}\sum_{m^{\prime}=1}^{\mathrm{M}_{3}}C_{m}^{(1)}(z)
C_{m^{\prime}}^{(6)}(\theta_{m}^{*})\mathbf{b}_{(k+1)2}^{+}(\xi_{m^{\prime}}),
\end{split}
\end{equation}
and
\begin{equation}\label{5.30}
\begin{split}
\mathbf{\Lambda}_{4}&=\sum_{m=1}^{\mathrm{M}_{3}}\sum_{m^{\prime}=1}^{\mathrm{M}_{2}}C_{m}^{(2)}(z)
C_{m^{\prime}}^{(9)}(\xi_{m}^{*})\mathbf{b}_{(k+1)2}^{+}(\theta_{m^{\prime}})
+\sum_{m=1}^{\mathrm{M}_{3}}\sum_{m^{\prime}=1}^{\mathrm{M}_{3}}C_{m}^{(2)}(z)
C_{m^{\prime}}^{(10)}(\xi_{m}^{*})\mathbf{b}_{(k+1)2}^{+}(\xi_{m^{\prime}}).
\end{split}
\end{equation}
Then, equations~\eqref{5.24},~\eqref{5.27} and~\eqref{5.28} comprise closed linear system of linear equations. These system may be written $(\mathbf{I}-\mathbf{P})\mathbf{x}_{k}=\mathbf{w}_{k}$, where $\mathbf{x}_{k}=(x_{k1},\cdots,x_{k(2\mathrm{M}_{1}+\mathrm{M}_{2}+\mathrm{M}_{3})})^{T}$ and
\begin{align}\label{5.31}
\begin{split}
x_{kj}=\left\{\begin{array}{ll}
\mathbf{b}_{(k+1)3}^{+}(z_{j}),  &j=1, \cdots, \mathrm{M}_{1}, \\
\mathbf{b}_{(k+1)1}^{-}(z_{j-\mathrm{M}_{1}}^{*}), & j=\mathrm{M}_{1}+1, \cdots, 2 \mathrm{M}_{1}, \\
\mathbf{b}_{(k+1)2}^{+}(\theta_{j-2\mathrm{M}_{1}}), & j=2\mathrm{M}_{1}+1, \cdots, 2 \mathrm{M}_{1}+\mathrm{M}_{2}, \\
\mathbf{b}_{(k+1)2}^{+}(\xi_{j-2\mathrm{M}_{1}-\mathrm{M}_{2}}), & j=2\mathrm{M}_{1}+\mathrm{M}_{2}+1, \cdots, 2 \mathrm{M}_{1}+\mathrm{M}_{2}+\mathrm{M}_{3},
\end{array}\right.
\end{split}
\end{align}
Using Cramer's rule, then
\begin{equation}\label{5.32}
\begin{split}
x_{k j}=\frac{\operatorname{det}\mathbf{K}_{kj}^{\text{aug}}}{\operatorname{det}\mathbf{K}},
\quad j=1, \cdots, 2\mathrm{M}_{1}+\mathrm{M}_{2}+\mathrm{M}_{3}, \quad k=1,2,
\end{split}
\end{equation}
where $\mathbf{K}_{kj}^{\text{aug}}=(\mathbf{K}_{1}, \cdots, \mathbf{K}_{j-1}, \mathbf{w}_{k}, \mathbf{K}_{j+1}, \cdots, \mathbf{K}_{2\mathrm{M}_{1}+\mathrm{M}_{2}+\mathrm{M}_{3}}) $. By substituting the obtained values into the reconstruction formula \eqref{4.19} and using the~\eqref{5.3} of $\eta_{j}(x, t)$, we can achieve the desired results.
\end{proof}

In the reflectionless case, the trace formulae have simpler expressions
\begin{subequations}\label{5.33}
\begin{align}
s_{11}(z)&=
\prod_{m=1}^{\mathrm{M}_{1}}\frac{z-z_{m}}{z-z_{m}^{*}}\frac{z+(q_{0}^{2}/z_{m}^{*})}{z+(q_{0}^{2}/z_{m})}
\prod_{m=1}^{\mathrm{M}_{2}}\frac{z+(q_{0}^{2}/\theta_{m})}{z+(q_{0}^{2}/\theta_{m}^{*})}
\prod_{m=1}^{\mathrm{M}_{3}}\frac{z-\xi_{m}}{z-\xi_{m}^{*}}, \quad z\in\mathbb{U}_{1}, \\
h_{22}(z)&=\mathrm{e}^{-\mathrm{i}\Delta\delta}\prod_{m=1}^{\mathrm{M}_{2}}\frac{z-\theta_{m}}{z-\theta_{m}^{*}}
\frac{z+(q_{0}^{2}/\theta_{m})}{z+(q_{0}^{2}/\theta_{m}^{*})}
\prod_{m=1}^{\mathrm{M}_{3}}\frac{z-\xi_{m}}{z-\xi_{m}^{*}}
\frac{z+(q_{0}^{2}/\xi_{m})}{z+(q_{0}^{2}/\xi_{m}^{*})}, \quad z\in\mathbb{U}^{+},
\end{align}
\end{subequations}
where $M_{1}$, $M_{2}$ and $M_{3}$ represent the number of three kinds of discrete eigenvalues respectively. The different combinations of the three types of eigenvalues are studied below. Subsequently, the solutions~\eqref{5.1} for parameter variations exhibit distinct wave structures. Then, we explicitly five wave structures for the solutions of the coupled LPD equation~\eqref{1.1} with NBCS.

\begin{figure}[htb]
    \centering
    \begin{tabular}{ccc}
\includegraphics[width=0.45\textwidth]{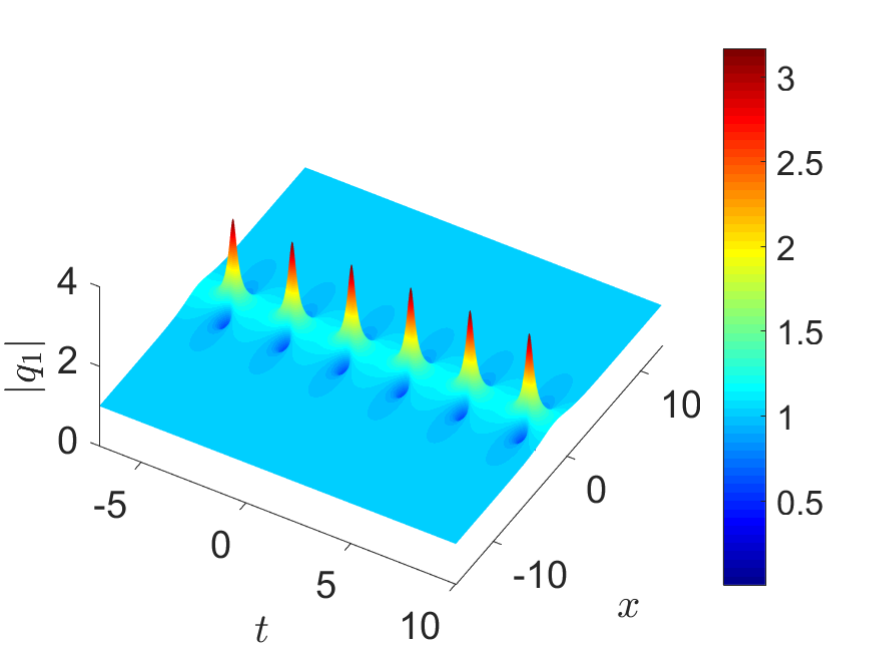} &
\includegraphics[width=0.45\textwidth]{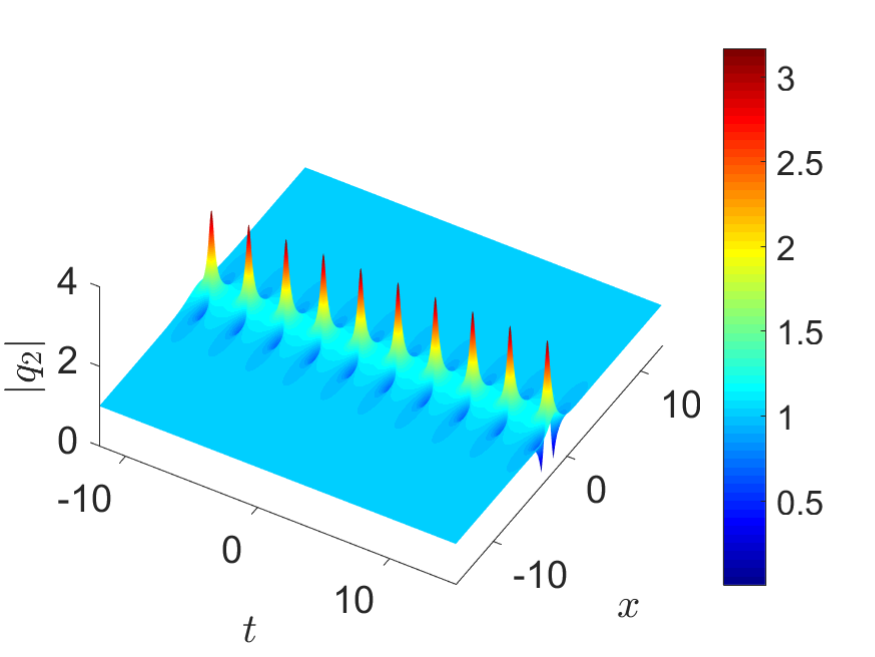} &
	\end{tabular}
\caption{One breather solution $q_{1}$ and $q_{2}$ by taking $\mathbf{q}_{+}=(\sqrt{2}/2,-\sqrt{2}/2)^{T},\sigma=1,d_{1}=\exp(-1.3+0.7\mathrm{i}),z_{1}=1.5\exp(\mathrm{i}\pi/2)$.
}
\label{fig:1}
\end{figure}

\subsection{Different types of soliton solutions}
\label{s:Pure soliton solutions:1}

Assuming $M_{1}=1$ and $M_{2}=M_{3}=0$ and considering
\begin{equation}\label{5.34}
\begin{split}
d_{1}=\exp(p_{1}+p_{2}\mathrm{i}), \quad  z_{1}=p_{3}\exp(\beta_{1}\mathrm{i}\pi), \quad
p_{1},p_{2},p_{3},\beta_{1}\in\mathbb{R},
\end{split}
\end{equation}
we utilize Theorem~\ref{thm:5} to derive one breather solution. From Fig.~\ref{fig:1}, one breather solution for two components $q_{1}$ and $q_{2}$ consist of periodically an upward peak and two downward valleys.

Assuming $M_{2}=1$ and $M_{1}=M_{3}=0$ and considering
\begin{equation}\label{5.38}
\begin{split}
g_{1}=\exp(p_{4}+p_{5}\mathrm{i}), \quad  \theta_{1}=p_{6}\exp(\beta_{2}\mathrm{i}\pi), \quad
p_{4},p_{5},p_{6},\beta_{2}\in\mathbb{R},
\end{split}
\end{equation}
we obtain one kink wave and one bright soliton solution. Fig.~\ref{fig:2} gives one kink wave solution $q_{1}$ and one bright soliton solution $q_{2}$.

\begin{figure}[htb]
\centering
\begin{tabular}{ccc}
\includegraphics[width=0.45\textwidth]{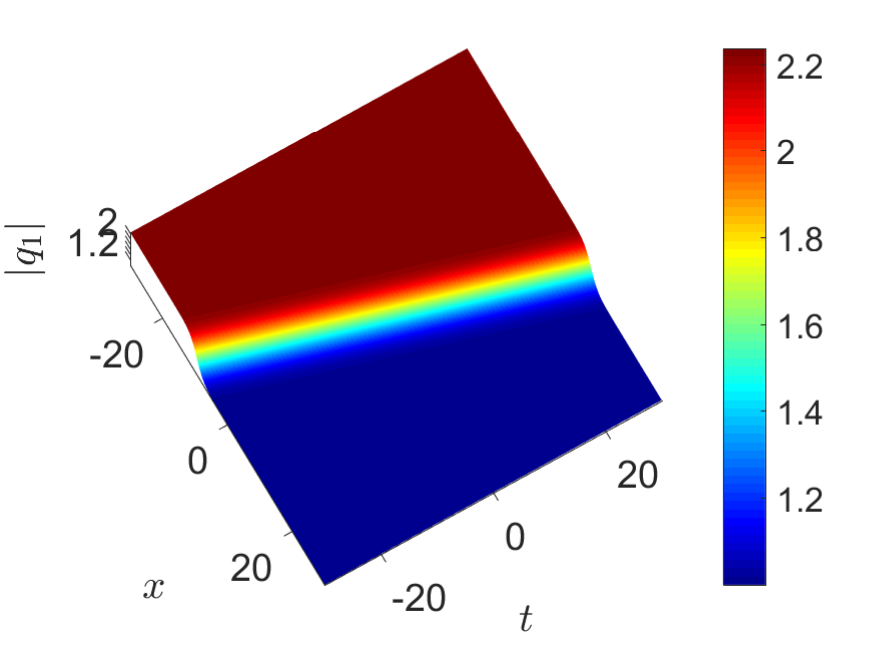} &
\includegraphics[width=0.45\textwidth]{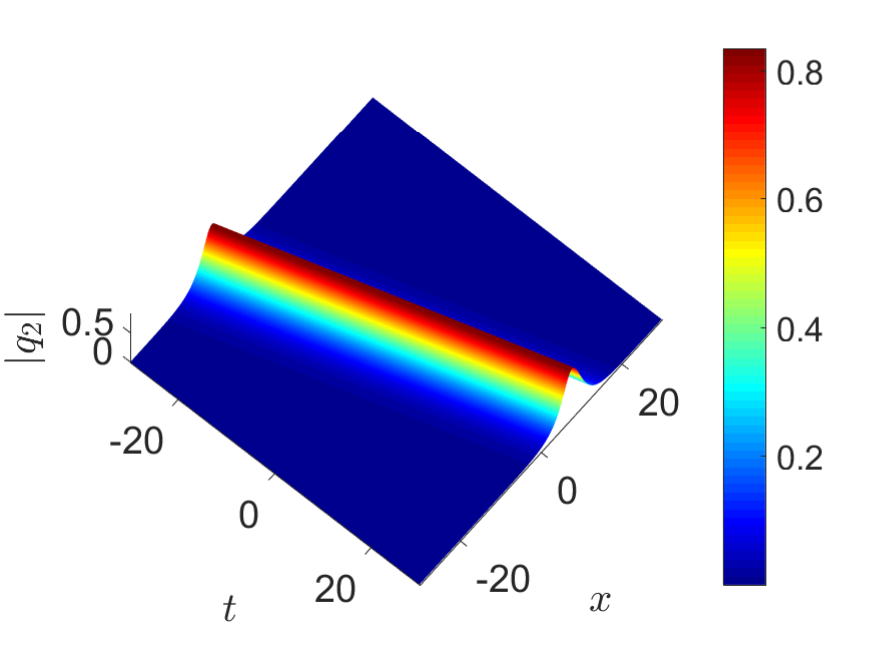} &
	\end{tabular}
\caption{One kink wave solution $q_{1}$ and one bright soliton solution $q_{2}$ by taking $\mathbf{q}_{+}=(1,0)^{T},\sigma=1,g_{1}=\exp(-0.6+\mathrm{i}),\theta_{1}=1.6\exp(\mathrm{i}\pi/4)$.
}
\label{fig:2}
\end{figure}

Assuming $M_{3}=1$ and $M_{1}=M_{2}=0$ and considering
\begin{equation}\label{5.41}
\begin{split}
f_{1}=\exp(p_{7}+p_{8}\mathrm{i}), \quad  \xi_{1}=p_{9}\exp(\beta_{3}\mathrm{i}\pi), \quad
p_{7},p_{8},p_{9},\beta_{3}\in\mathbb{R},
\end{split}
\end{equation}
we derive one dark soliton and one bright soliton solution. Fig.~\ref{fig:3} shows one dark soliton solution $q_{1}$ and one bright soliton solution $q_{2}$.

\begin{figure}[htb]
\centering
\begin{tabular}{ccc}
\includegraphics[width=0.45\textwidth]{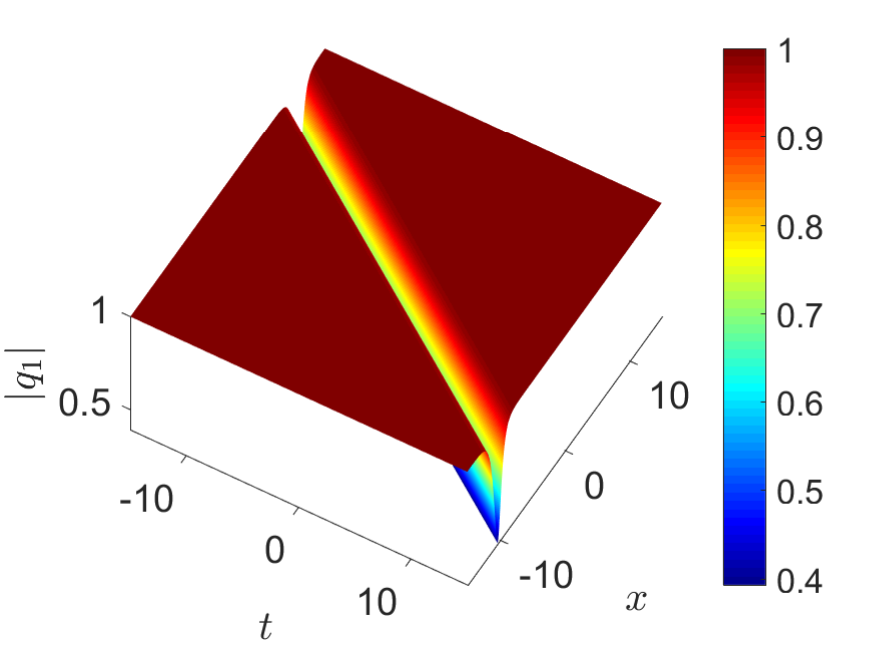} &
\includegraphics[width=0.45\textwidth]{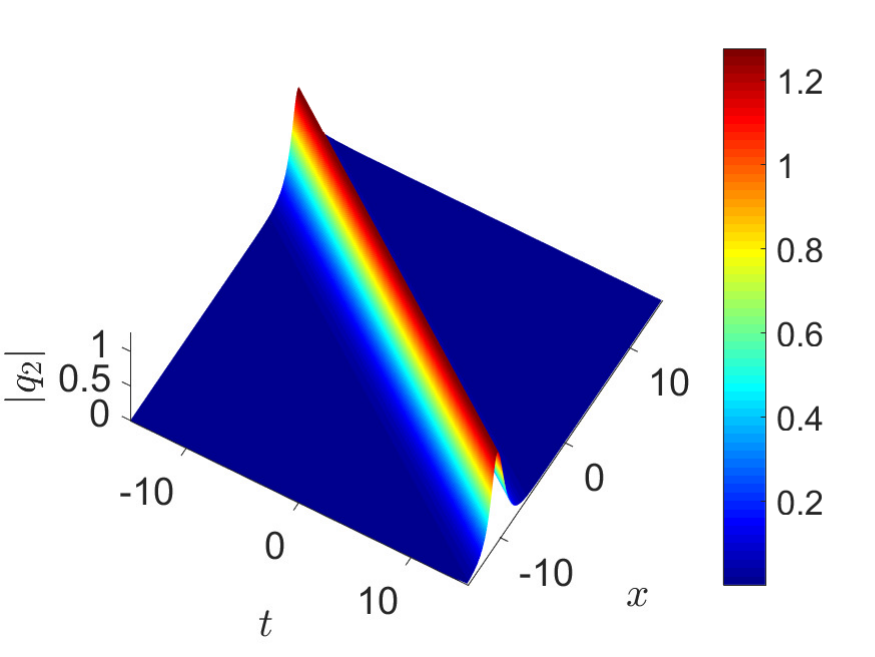} &
	\end{tabular}
\caption{One dark soliton solution $q_{1}$ and one bright soliton solution $q_{2}$ by taking $\mathbf{q}_{+}=(1,0)^{T},\sigma=1,f_{1}=\exp(0.4+\mathrm{i}),\xi_{1}=1.5\exp(\mathrm{i}\pi/4)$.
}
\label{fig:3}
\end{figure}

Assuming $M_{1}=0$ and $M_{2}=M_{3}=1$ and considering
\begin{equation}\label{5.44}
\begin{split}
g_{1}&=\exp(p_{4}+p_{5}\mathrm{i}), \quad  \theta_{1}=p_{6}\exp(\beta_{2}\mathrm{i}\pi), \quad
p_{4},p_{5},p_{6},\beta_{2}\in\mathbb{R}, \\
f_{1}&=\exp(p_{7}+p_{8}\mathrm{i}), \quad  \xi_{1}=p_{9}\exp(\beta_{3}\mathrm{i}\pi), \quad
p_{7},p_{8},p_{9},\beta_{3}\in\mathbb{R},
\end{split}
\end{equation}
we find the hybrid solution between one kink wave and one dark soliton, two bright solitons solution. In Fig.~\ref{fig:4}, $q_{1}$ demonstrates the interaction of one kink wave and one dark soliton. In Fig.~\ref{fig:4}, $q_{2}$ analyzes the collision between two bright solitons. Hybrid solution between one kink wave and one dark soliton $q_{1}$ and two bright solitons soliton $q_{2}$ are described by three-dimensional graph.

\begin{figure}[htb]
\centering
\begin{tabular}{ccc}
\includegraphics[width=0.45\textwidth]{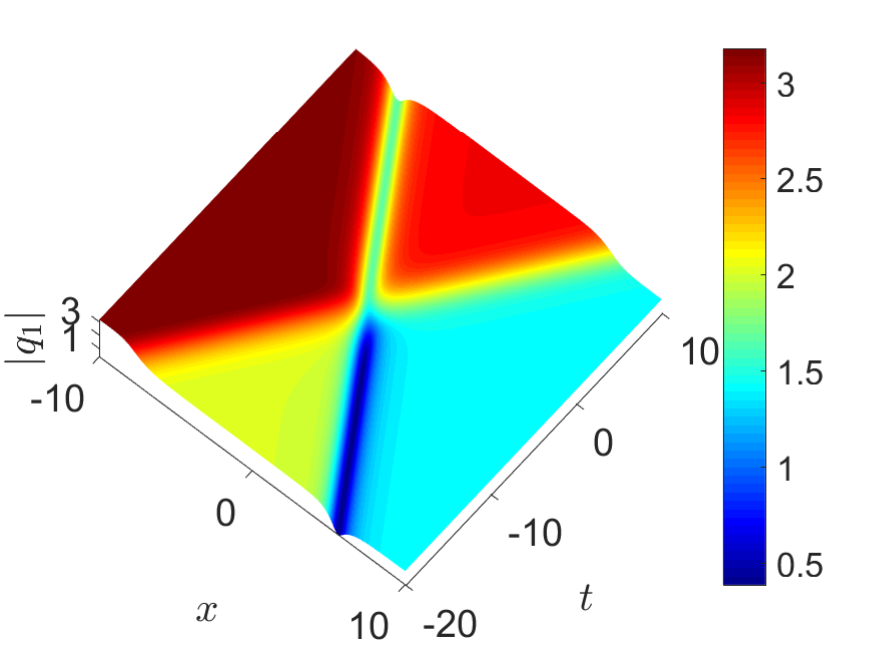} &
\includegraphics[width=0.45\textwidth]{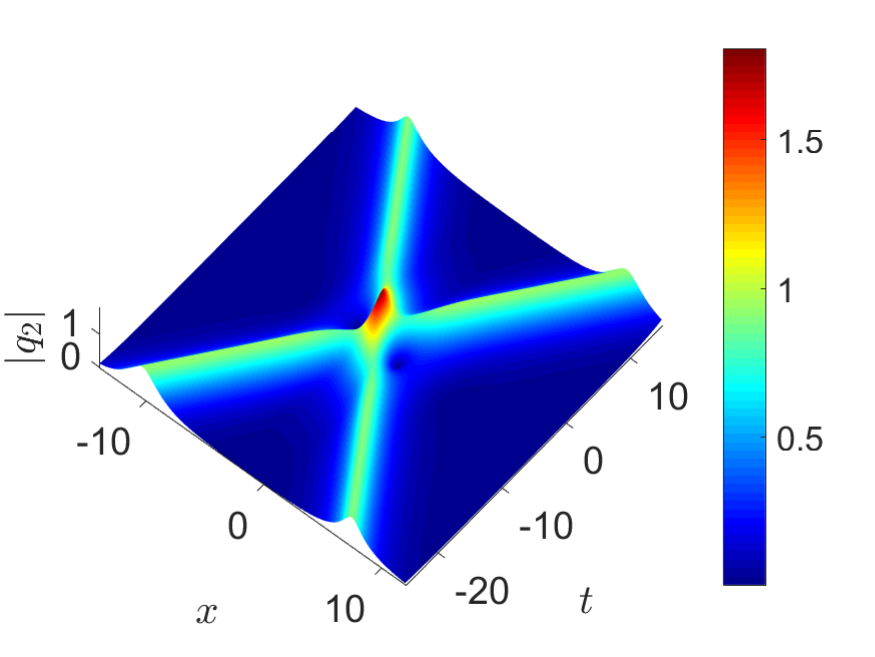} &
	\end{tabular}
\caption{Hybrid solution between one kink wave and one dark soliton $q_{1}$ and two bright solitons solution $q_{2}$ by taking $\mathbf{q}_{+}=(1,0)^{T},\sigma=1,g_{1}=\exp(1+\mathrm{i}),
\theta_{1}=1.4\exp(\mathrm{i}\pi/4),f_{1}=\exp(1.8+\mathrm{i}),\xi_{1}=0.9\exp(-\mathrm{i}\pi/4)$.
}
\label{fig:4}
\end{figure}

Assuming $M_{1}=M_{2}=M_{3}=1$ and considering
\begin{equation}\label{5.49}
\begin{split}
d_{1}&=\exp(p_{1}+p_{2}\mathrm{i}), \quad g_{1}=\exp(p_{4}+p_{5}\mathrm{i}), \quad f_{1}=\exp(p_{7}+p_{8}\mathrm{i}), \quad  \beta_{1},\beta_{2},\beta_{3}\in\mathbb{R}, \\
z_{1}&=p_{3}\exp(\beta_{1}\mathrm{i}\pi), \quad \theta_{1}=p_{6}\exp(\beta_{2}\mathrm{i}\pi), \quad
\xi_{1}=p_{9}\exp(\beta_{3}\mathrm{i}\pi), \quad p_{i}\in\mathbb{R} (i=1,\cdots,9),
\end{split}
\end{equation}
we obtain the hybrid solution between one breather, one bright soliton and one dark soliton, hybrid solution between one breather and two dark solitons. In Fig.~\ref{fig:5}, $q_{1}$ demonstrates the interaction of one breather, one bright soliton and one dark soliton. In Fig.~\ref{fig:5}, $q_{2}$ analyzes the collision between one breather and two dark solitons. Hybrid solution between one breather, one bright soliton and one dark soliton $q_{1}$ and hybrid solution between one breather and two dark solitons $q_{2}$ are described by three-dimensional graph.

\begin{figure}[htb]
\centering
\begin{tabular}{ccc}
\includegraphics[width=0.45\textwidth]{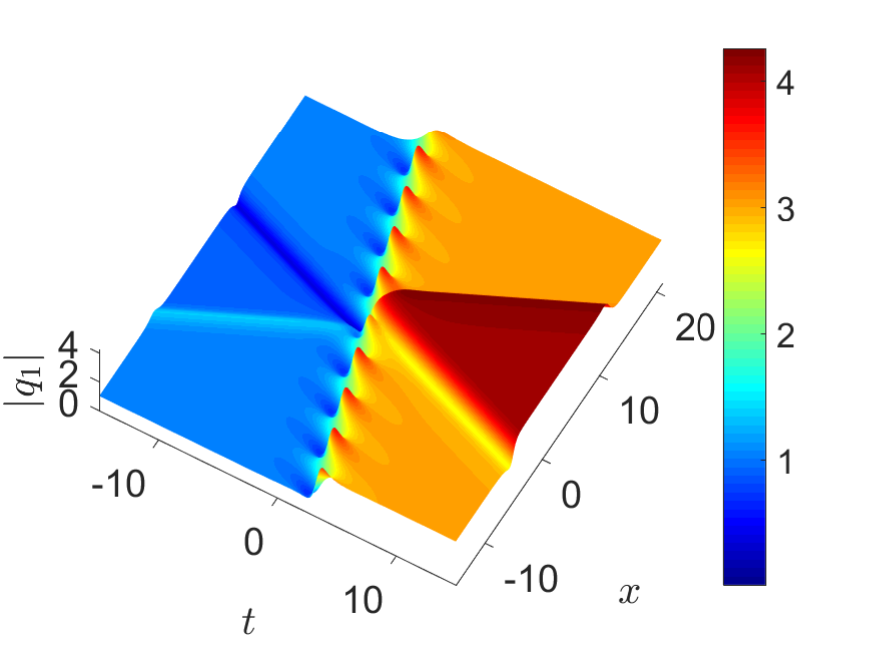} &
\includegraphics[width=0.45\textwidth]{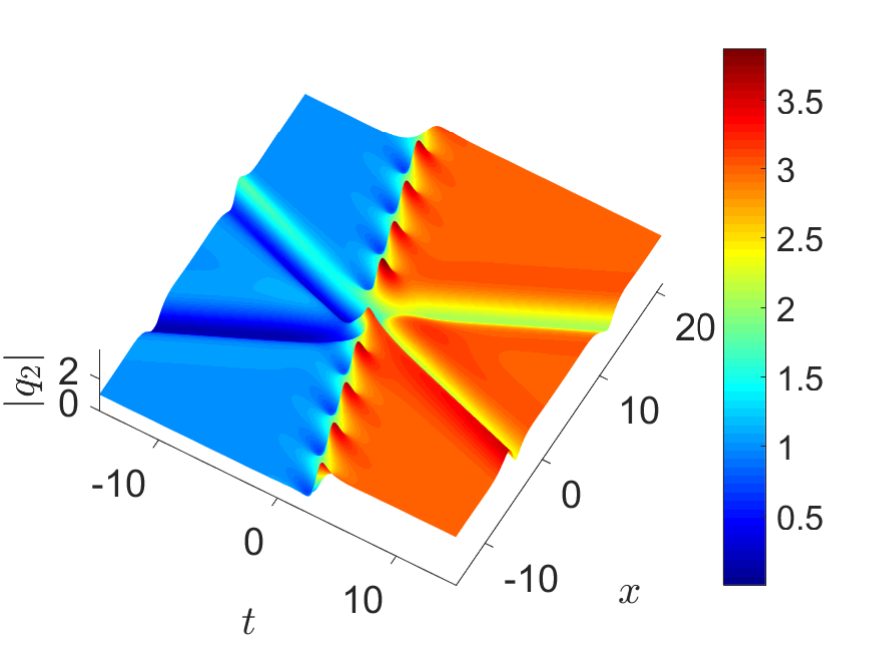} &
	\end{tabular}
\caption{Hybrid solution between one breather, one bright soliton and one dark soliton $q_{1}$ and the hybrid solution between one breather and two dark solitons $q_{2}$ by taking $\mathbf{q}_{+}=(\sqrt{2}/2,-\sqrt{2}/2)^{T},\sigma=1,d_{1}=\exp(-0.1+2.3\mathrm{i}),z_{1}=1.4\exp(2.26\mathrm{i}\pi),g_{1}=\exp(2.3-\mathrm{i}),
\theta_{1}=1.2\exp(\mathrm{i}\pi/4),f_{1}=\exp(1.3-0.25\mathrm{i}),\xi_{1}=\exp(\mathrm{i}\pi/4)$.
}
\label{fig:5}
\end{figure}

\section{Conclusion}
\label{s:Conclusion}

We have developed the IST for the coupled LPD equation~\eqref{1.1} with NZBC and obtained some new results. The analyticity of the Jost eigenfunctions and scattering coefficients is explored, along with specific potential conditions ensuring this analytic behavior. Novel analytical eigenfunctions for the coupled LPD equation satisfy two symmetry relations, which are then employed to rigorously characterize the discrete spectrum. The discrete spectrum produces discrete eigenvalues in three different cases, corresponding to a variety of different types of soliton solutions, and their soliton interaction characteristics are illustrated graphically.

In recent years, there has been a surge of captivating new applications of the IST in various experimental scenarios. The IST is employed to address the issue of obtaining analytical solutions for the nontrivial background of the focusing NLS equation on significant statistical rogue wave events. A novel rogue wave classification approach is introduced, and it is analyzed using numerical IST techniques~\cite{E1}. The IST was used to exploit exact analytical tools for generating dark soliton arrays on-demand~\cite{E2}. Many scholars have conducted more extensive research on the IST~\cite{E3,E4,E5}. However, the study of the rogue wave solutions of the coupling equation by means of the IST has not been studied, so the study of this problem can be considered.

The RH problem has proven to be highly efficient in generating soliton solutions. Recently, it has been expanded to tackle a wide range of initial-boundary value problems for continuous integrable equations on half-lines and finite intervals~\cite{E6}. At the same time, the RH problem of non-parallel boundary conditions at infinity~\cite{E7} can also be considered. It is crucial to establish the link between various methods in order to demonstrate the dynamic properties of soliton solutions. It is another interesting topic for future study to link RH problems to generalized integrable counterparts, including integrable couplings, super hierarchies, super-symmetric integrable equations and fractional spacetime analogous equations.

Practically, we anticipate that the findings in this paper will have utility in understanding recent experiments~\cite{H8} involving plasma physics, fluids, solid-state physics, plasmas, quantum mechanics and Bose-Einstein condensation. Future directions in this area might explore extending the IST to study significant boundary conditions in mathematical and physical problems, such as more complex boundary conditions that depend on the spatial variable. Furthermore, we hold an even more ambitious perspective that the forefront of our exploration resides in extending the scope of the IST to encompass potentials exhibiting gradual decay. These conditions are intricately linked to the evolving framework of breather and rogue wave solutions, marking a significant advancement in our emergent theory. It has witnessed a significant revival, the prosperity of papers and new exciting applications. This trend bodes positively for the times ahead.

\section*{Acknowledgement}
This work was supported by the National Natural Science Foundation of China (Grant Nos. 11371326, 11975145 and 12271488).

\section*{Conflict of interests}
The authors declare that there is no conflict of interests regarding the research effort and the publication of this paper.

\section*{Data availability statements}
All data generated or analyzed during this study are included in this published article.

\section*{ORCID}
{\setlength{\parindent}{0cm}
Peng-Fei Han: https://orcid.org/0000-0003-1164-5819 \\
Yi Zhang: https://orcid.org/0000-0002-8483-4349}

References



\end{document}